\documentclass[a4paper]{article} 

\usepackage{amsmath}
\usepackage{amsthm}
\usepackage{breakcites}
\usepackage{nicefrac}
\usepackage{multicol}
\usepackage{multirow}
\usepackage{thmtools} 
\usepackage{thm-restate}
\usepackage[hidelinks]{hyperref}
\usepackage{booktabs}
\usepackage{caption}
\usepackage{subcaption}

\usepackage{tikz,pgf}
\usetikzlibrary{shapes,backgrounds,shapes.geometric,shadows,patterns,positioning,arrows}
\tikzstyle{arg}=[draw,circle,fill=gray!15,inner sep=1pt,minimum size=.5cm]
\tikzstyle{attack}=[->,left,thick,>=stealth]
\tikzstyle{cond} = [rectangle, inner sep=5pt]

\newcommand{\ccfont}[1]{\protect\mathsf{#1}}
\newcommand{\Polytime}{\ccfont{P}}
\newcommand{\NP}{\ccfont{NP}}

\newcommand{\ThetaP}[1]{\Theta_{#1}\Polytime}

\newcommand{\DeltaParallel}[1]{\Polytime_{\parallel[{#1}]}^\NP}

\newtheorem{example}{Example}
\newtheorem{theorem}{Theorem}

\newtheorem{definition}{Definition}
\newtheorem{observation}[theorem]{Observation}

\newcommand{\af}{F}
\newcommand{\absaf}{\mathcal{S}}
\newcommand{\args}{\mathit{Arg}}
\newcommand{\att}{\mathit{Att}}
\newcommand{\voters}{N}
\newcommand{\ballots}{{\bar A}}
\newcommand{\ballot}[1]{A(#1)}

\newcommand{\cf}{\mathit{cf}}

\newcommand{\adm}{\mathit{adm}}
\newcommand{\pref}{\mathit{prf}}

\newcommand{\outcome}{\Omega}
\newcommand{\pov}{\pi}
\newcommand{\maxCov}{\text{MaxCov}}
\newcommand{\weight}{w}
\newcommand{\OWA}{\text{OWA}_{\vec{\weight}}}
\newcommand{\greedOWA}{\text{GreedOWA}_{\vec{\weight}}}
\newcommand{\greedCov}{\text{GreedCov}}
\newcommand{\rep}{\mathit{rep}}
\newcommand{\repDef}{\rep^c}
\newcommand{\maxDef}[1]{\mu(#1)}

\DeclareMathOperator*{\argmax}{arg\,max}

\theoremstyle{plain}

\usepackage{natbib}
\setcitestyle{numbers,square}

\usepackage{kpfonts}
\usepackage[T1]{fontenc}
\usepackage[hidelinks]{hyperref}
\usepackage[nocompress]{cite}

\title{Combining Voting and Abstract Argumentation to Understand Online Discussions}
\date{\small DBAI, TU Wien\\\texttt{\{firstname.lastname\}@tuwien.ac.at}}

\author{
  Michael Bernreiter
  \hspace{3cm}
  Jan Maly
  \\[3ex]
  Oliviero Nardi
  \hfill
  Stefan Woltran
}
         
\newcommand{\BibTeX}{\rm B\kern-.05em{\sc i\kern-.025em b}\kern-.08em\TeX}

\begin{document}
\sloppy  

\maketitle 

\begin{abstract}
	Online discussion platforms are a vital part of the public discourse
	in a deliberative democracy. However, how to interpret
	the outcomes of the discussions on these platforms is often unclear.
	In this paper, we propose a novel and explainable method
	for selecting a set of most representative, consistent points of view
	by combining methods from computational social choice and abstract argumentation.
	Specifically, we model online discussions as abstract argumentation frameworks
	combined with information regarding which arguments voters approve of.
	Based on ideas from approval-based multiwinner voting,
	we introduce several voting rules for selecting a set of preferred extensions
	that represents voters' points of view.
	We compare the proposed methods across several dimensions, theoretically
	and in numerical simulations, and give clear suggestions on which methods to
	use depending on the specific situation.
\end{abstract}


\section{Introduction}

In recent years, a veritable ``deliberative wave'' has swept through
many democratic societies \citep{OECD}, bringing with it many new discursive 
participation formats, from citizen assemblies to online discussion platforms. 
These formats allow citizens to discuss (often very divisive) political issues
and thus can enable us to understand which opinions are held by well-informed citizens.
In this paper, we focus on text-based online discussion platforms used to inform political decision making, such as \emph{Polis} (\url{pol.is}),
\emph{Your Priorities} (\url{yrpri.org}) or \emph{Decidim} (\url{decidim.org})
to name just a few of the rapidly expanding number of tools used around the world \citep{SDI}.
These platforms allow citizens to submit comments, but also to approve (and sometime
disapprove) comments others have posted. While such platforms are easy to set up and use,
they have the downside that discussions can be unorderly and chaotic, which makes it hard
to interpret and summarize them. One approach for solving this problem -- used for example
by \emph{Polis} \citep{small2023opportunities} -- is to use machine learning
and statistics to find clusters of voters 
with similar opinions, as well as a set of comments representing their joint opinion.
However, with this method the comments used to represent a cluster of 
voters might be inconsistent. Additionally, the obtained results are generally
not explainable, which is highly problematic for processes that inform political decision making.

In this paper, we propose an explainable method for picking consistent points of view
that represent the opinions of voters using methods from computational social choice (ComSoC)
combined with tools from abstract argumentation~\citep{Dung1995}. Picking a representative set of
most popular comments based on citizens' approvals is an aggregation problem
that is closely related to the well-studied setting of
approval-based multiwinner voting \citep{abcbook},
for which many voting rules guaranteeing fair representation have been introduced lately \citep{aziz2017justified,peters2020proportionality}.
However, as every comment is only 
one argument in a larger discussion, we would argue that
it does not suffice to pick a single representative set
of popular comments, but that we have to find \emph{sets} of comments that represent
citizens' points of view.

Moreover, these points of view should be \emph{consistent} if we want them to
influence political decision making. To guarantee this, we need additional semantic
information about the comments. Depending on the structure of the discussion, different 
formalisms could be used to capture this information. We focus on 
abstract argumentation frameworks (AFs), one of the most well-studied formalisms
for representing conflicting arguments.
AFs are particularly well-suited for our setting due to their
minimal and easy to understand syntax that only requires the specification of attacks
between different comments or arguments.
These attack relations
either have to be added by a moderator,
crowd-sourced from the participants, or mined from natural language text using argument mining techniques ~\citep{PalauM09a,CabrioV12,GoffredoCVHS23,BaumannWHHH20}.

Our contributions are as follows: 
we formally introduce \emph{Approval-Based Social Argumentation Frameworks} to model online discussions in which citizens can 
approve arguments.
We study the problem of selecting a small but representative set of so-called preferred extensions,
which are maximal consistent and defendable sets of arguments. 
We study this problem from two sides.
First, in smaller discussions, we might want to pick an, ideally small, set 
of extensions that represents every voter perfectly. Whether this is possible depends, as we shall show,
on how we conceptualize perfect representation.
Secondly, for large discussions, we might be more interested in 
picking a small number of extensions that represent the voters \emph{as well as possible}.
We propose several 
methods for doing so, based on well-known voting rules
from the ComSoC literature, and compare them with respect to their computational complexity,
their axiomatic properties and their performance in numerical simulations. 
Based on these results, we make precise recommendations for which methods are best suited for
different applications. 

The code \citep{ourcode} used in our experiments is available online.



\subsection*{Related Work}

Social choice theory has been used before to analyze discussions in participation platforms (like \emph{Polis}), e.g., by \citet{halpern2023representation} 
and \citet{fish2023generative} 
However, neither of these works employs argumentation to capture the relationship between comments, nor are they concerned with choosing a consistent subset of them. Hence, they are technically quite  distinct from ours.
Indeed, to the best of our knowledge, our paper is the first to consider voter representation in abstract argumentation. 

That said, there is already a significant amount of literature on combining argumentation and social choice \citep{baumeister2021collective}. One strain of literature \citep{caminada2011judgment,AwadEtAl2017} 
applies ideas from judgment aggregation \citep{EndrissHBCOMSOC2016} to argumentation. Here, the objective is to aggregate individual judgments about the acceptability of arguments. Crucially, these works assume that voters submit ballots adhering to strict rationality constraints, which is not realistic in online discussions, where voters are often not even aware of all arguments at the time of voting.

Another related approach is that of \emph{weighted} AFs~\citep{dunne2011weighted,BistarelliSantini2021}, where attacks between arguments have weights that can be established by letting agents vote on arguments, similar to our scenario. Similarly, in \emph{social} AFs \citep{LeiteMartins2011}, voters can approve (or disapprove) arguments. The strength of an argument is then related to its social support and the support of its attacking arguments. Beyond abstract argumentation, in the \emph{relational reasoning model}~\citep{ganzer2023model} voters can judge the acceptability of a set of arguments and their relationships (by assigning them weights). 
The goal here is to aggregate such judgments and collectively evaluate a set of target arguments.
While there are conceptual similarities to our approach, all three formalisms redefine argumentation semantics in light of votes (e.g.\ by considering the strength of arguments or attacks). In contrast, we rely on standard semantics for AFs
and use the additional approval information to select a representative subset of extensions.
Most importantly, none of the papers mentioned above study the problem of voter representation, which is the main objective of our paper.

Finally, there exists a significant body of work focused on the merging of AFs \citep{COSTEMARQUIS2007730,tohme2008aggregation,dunne2012argument,delobelle2016merging,ChenEndrissAIJ2019}. Here, each agent is endowed with a framework, and the goal is to merge these into a collective one. This differs significantly from our approach, where the AF is given and we have to decide which extension to select.

\section{Preliminaries}


The basic problem considered in this paper,
namely selecting representative comments based on the approvals of 
voters, is an aggregation problem with the following components:
a finite set of voters $\voters$, a finite set of candidates $C$, and a vector
of approval ballots $\ballots=\left(\ballot{i}\right)_{i\in \voters}$,
where $\ballot{i} \subseteq C$. Note that this is equivalent to 
the input of an approval-based single- or multiwinner election. 
The difference to these formalisms lies in the outcome $W$ we try to select:
we want to select a set $W$ of subsets of $C$, i.e., $W \subseteq 2^C$
where $2^C$ denotes the power
set of $C$. Often, we will study the problem of selecting exactly $k$ subsets,
i.e., we will stipulate $|W| = k$. However, in contrast 
to most multiwinner voting settings, we will not constrain the cardinality of the 
selected subsets in $W$. Instead, we impose consistency constraints using abstract 
argumentation.

\subsubsection*{Argumentation}

AFs~\citep{Dung1995} are a well-studied formalism in which discussions can be represented and reasoned about. 
Arguments (denoted by lower-case letters $a,b,c,\ldots$) in AFs are abstract entities, i.e., we are not concerned with their internal structure but rather with the relationship between them. 
Specifically, an argument $x$ can attack another argument $y$, which implies that $x$ and $y$ are in conflict and cannot be jointly accepted. 
To accept $y$, it must be defended against $x$'s attack, i.e.,
either $y$ attacks $x$ or there is another argument $z$ which attacks $x$ and can be accepted alongside $y$.

\begin{definition} \label{def:AF}
  An argumentation framework (AF) is a tuple $\af = (\args,\att)$ where $\args$ is a finite set of arguments and $\att \subseteq \args \times \args$ is an attack relation between arguments. 
  Let $S \subseteq \args$. $S$ \emph{attacks} $b$ (in~$\af$) if $(a,b)\in \att$ for some $a\in S$; 
  $S^+_F=\{b\in \args \mid\exists a\in S: (a,b)\in \att\}$ denotes the set of arguments attacked by $S$.
  An argument $a \in \args$ is \emph{defended} (in $\af$) by $S$ if $b\in S^+_F$ for each $b$ with $(b,a)\in \att$.
\end{definition}

AFs are usually visualized as directed graphs, where a node is an argument and an edge an attack between arguments (see Figure~\ref{fig:example-absaf}). 

AF-semantics are 
functions $\sigma$ that assign 
a set
$\sigma(\af)\subseteq 2^{\args}$
of extensions
to each AF
$\af=(\args,\att)$.
Conflict-free ($\sigma = \cf$) semantics select sets $S \subseteq \args$ where no two arguments attack each other. Admissible semantics ($\sigma = \adm$) select conflict-free sets that defend themselves. 
Preferred semantics ($\sigma = \pref$) select subset-maximal admissible sets.  
Many alternative AF-semantics have been defined~\citep{BaroniEtAl2018}, but we focus on the well-established preferred semantics.

\begin{definition}\label{def:semantics}
  Let $\af = (\args,\att)$ be an AF. 
  For $S\subseteq \args$ it holds that
  \begin{itemize}
    \item $S \in \cf(\af)$ iff there are no $a, b \in S$ such that $(a,b) \in \att$;
    \item $S\in\adm(\af)$ 
    iff $S \in \cf(\af)$ and each $a\in S$ is defended by $S$; 
    \item $S\in\pref(\af)$ 
    iff $S\in\adm(\af)$ and $S \not\subset T$ for all $T\in\adm(\af)$.
  \end{itemize}
\end{definition}

\begin{example} \label{example:af}
	Let $\af$ be the AF from Figure~\ref{fig:example-absaf}. 
	Note that $\{p_1,p_2\} \in \cf(\af)$ but $\{f_1,p_2\} \not\in \cf(\af)$.
	Regarding admissible semantics, $\{p_2\} \in \adm(\af)$ since $p_2$ defends itself against the attack $(f_1,p_2)$. 
	However, $\{p_1\} \not\in \adm(\af)$ since $p_1$ does not defend itself against $(f_1,p_1)$, while 
	$\{p_1,p_2\} \in \adm(\af)$ since $p_2$ defends $p_1$. 
	There are~8 preferred extensions, namely 
	\begin{align*}
		\pref(F) = \{ 
		& \{p_1,p_2, \allowbreak p_3,s_1,s_2\},\! \allowbreak
		\{p_1,p_2, \allowbreak p_3,m_1\},\! \allowbreak
		\{f_2, p_1,p_2, \allowbreak s_1,s_2\},\! \allowbreak
		\{f_2, p_1,p_2, \allowbreak m_1\},\! \allowbreak \\
		& \{f_1,p_3, \allowbreak s_1,s_2\},\! \allowbreak
		\{f_1, \allowbreak p_3,m_1\},\! \allowbreak
		\{f_1,f_2, \allowbreak s_1,s_2\},\! \allowbreak
		\{f_1,f_2,m_1\}
		\}.
	\end{align*}
\end{example}

\subsubsection*{Complexity Theory}

We assume familiarity with complexity classes $\Polytime$ and $\NP$. 
Moreover, $\ThetaP{2}$ is the class of decision problems solvable in polynomial time with access to $O(\log n)$ $\NP$-oracle calls~\citep{Wagner1990}.
%

\begin{figure}[t]
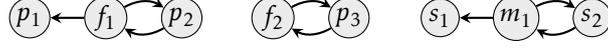

	\centering
	\tikz{
		\node[arg] (a21) at (0,0) {$p_1$};
		\node[arg] (a16) at (1,0) {$f_1$};
		\node[arg] (a42) at (2,0) {$p_2$};
		\node[arg] (a15) at (3.2,0) {$f_2$};
		\node[arg] (a43) at (4.2,0) {$p_3$};
		\node[arg] (a162) at (5.4,0) {$s_1$};
		\node[arg] (a163) at (6.4,0) {$m_1$};
		\node[arg] (a168) at (7.4,0) {$s_2$};
		\draw[attack] 
		(a16) edge (a21)
		(a16) edge[bend left] (a42)
		(a42) edge[bend left] (a16)
		(a15) edge[bend left] (a43)
		(a43) edge[bend left] (a15)
		(a163) edge (a162)
		(a163) edge[bend left] (a168)
		(a168) edge[bend left] (a163)
		;
	}
	\caption{AF for the Canadian election reform discussion.}
	\label{fig:example-absaf}
\end{figure}
\setlength{\tabcolsep}{0.4em}

\begin{figure}[t]
	\begin{minipage}{0.92\textwidth}
		\textbf{33}$\times \{p_1\}$;
		\textbf{31}$\times\{p_1,p_2,p_3\}$;  
		\textbf{16}$\times\{p_2\}$,$\{p_3\}$;
		\textbf{11}$\times\{f_2\}$;
		\textbf{10}$\times\{p_2,p_3\}$;  
		\textbf{9}$\times\{p_1,p_3\}$;  
		\textbf{9}$\times\{f_2,p_1\}$;  
		\textbf{8}$\times\{p_1,p_2\}$;  
		\textbf{7}$\times\{f_2,p_1,p_2,p_3\}$;
		\textbf{6}$\times\{s_2\}$;
		\textbf{4}$\times\{f_1,f_2\}$, $\{p_1,p_2,p_3,m_1,s_1\}$;
		\textbf{3}$\times\{f_1\}$, $\{f_2,p_2,p_3\}$;
		\textbf{2}$\times\{f_2,p_1,p_3\}$, $\{f_1,p_1\}$, $\{f_2,p_1,p_2\}$, $\{p_1,p_2,p_3,s_1\}$;\\
		\textbf{1}$\times\{p_2,m_1\}$, $\{s_1\}$, $\{m_1\}$, $\{f_2,p_3\}$, $\{p_2,p_3,s_2\}$,  $\{f_1,f_2,p_3\}$, $\{p_2,s_1\}$, $\{p_1,p_2,p_3,s_1,s_2\}$, $\{f_1,p_1,p_2,p_3\}$.
	\end{minipage}
	\caption{Approval ballots (Canadian election reform data).}
	\label{tab:ballots-canada}
\end{figure}

\section{Approval-based Social AFs} \label{sec:absaf}

Let us now introduce our main object of study, \emph{Approval-Based Social AFs} (ABSAFs),
which model discussions where agents can approve arguments that they find convincing. 

\begin{definition} \label{def:absaf}
An ABSAF $\absaf=(\af,\voters,\ballots)$ consists of an AF $\af = (\args,\att)$, a finite set $\voters$ of agents (also called voters), and
a vector $\ballots=\left(\ballot{i}\right)_{i\in \voters}$ of approval ballots where, for every agent $i \in \voters$, $\ballot{i} \subseteq \args$ with $\ballot{i} \neq \emptyset$ is the set of
arguments approved by $i$. 
\end{definition}

Throughout the paper, we denote the number of voters by $n$ and identify each voter by a natural number $i$, i.e., $\voters = \{1,\ldots,n\}$. 

We do not stipulate any constraints for the
submitted ballots -- not even conflict-freeness --
as ballots containing conflicting arguments appear in real-world examples
(see Example~\ref{example:absaf}).
The goal is to select a (usually small) set of coherent viewpoints representing the agents. 

\begin{definition} \label{def:outcome}
	An outcome $\outcome \subseteq \sigma(\af)$ of an ABSAF $\absaf=(\af,\voters,\ballots)$ is a set of $\sigma$-extensions. We call $\pov \in \outcome$ a point of view (or viewpoint).
\end{definition}

\begin{example} \label{example:absaf} 	
  We consider a discussion on an election reform proposed in Canada, taken from the ``Computational Democracy Project''.\footnote{\url{https://compdemocracy.org/}}
  For succinctness, we show two arguments in favor of the first-past-the-post (FPTP) electoral system ($f_1,f_2$), three arguments in favor of proportional representation (PR) ($p_1,p_2,p_3$), two arguments in favor of the single transferable vote (STV) system ($s_1,s_2$), and one argument in favor of a mixed-member proportional (MMP) system ($m_1$). The ``\textit{cid}'' code refers to the comment-id in the original dataset.\footnote{\url{https://github.com/compdemocracy/openData/tree/master/canadian-electoral-reform}}
  \begin{itemize}
    \item $f_1$ (\textit{cid} 16): 
    ``In systems with PR, opinions are too strongly divided and nothing gets done.''
    
    \item $f_2$ (\textit{cid} 15): 
    ``FPTP results in more stable governance.''
    
    \item $p_1$ (\textit{cid} 21): 
    ``A party’s share of seats in the House of Commons should reflect its share of the popular vote.''
    
    \item $p_2$ (\textit{cid} 42): 
    ``PR will reduce hyper-partisanship and promote cooperation between parties.''
    
    \item $p_3$ (\textit{cid} 43): 
    ``PR will reduce wild policy swings and result in more long-lasting policies.''
    
    \item $s_1$ (\textit{cid} 162): ``STV's advantage over MMP is that it doesn't explicitly enshrine political parties in our electoral system.''
    
    \item $s_2$ (\textit{cid} 168): 
    ``I like the simplicity of stating my favorite candidate, as well as alternative choices.''
    
    \item $m_1$ (\textit{cid} 163): ``MMP has a better chance of obtaining public support than STV due to its relative simplicity.''
  \end{itemize}
  We manually added attacks between arguments.
  The resulting AF $\af = (\args,\att)$ is shown in Figure~\ref{fig:example-absaf} and is the same as in Example~\ref{example:af}. 
  
  Table~\ref{tab:ballots-canada} shows the ballots along with how many agents voted for the given ballot, as extracted from the original dataset.
  Most voters approve arguments in favor of PR. Some voters approve both arguments in favor of PR and FPTP, e.g., the ballot $\{f_1,p_2\}$. Note that this ballot is conflict-free. However, not all ballots are conflict-free, e.g., $\{f_2,p_1,p_2,p_3\}$ was approved by seven voters.
  Moreover, the ballot with the most voters, namely $\{p_1\}$, is not admissible. 
\end{example}

Our main concern will be to find an outcome that is small, in that it contains
few points of view, but also represents as many agents as possible.
For this, we need to formally define what it means for an agent to be represented by a point of view or by an outcome. 

\begin{definition} \label{def:representation}
Let $\absaf=(\af,\voters,\ballots)$ be an ABSAF and let $\alpha\in[0\dots 1]$.
A point of view $\pov \in \sigma(\af)$ $\alpha$-represents voter $i \in \voters$ iff
\[\rep_i(\pov)=\frac{|\pov \cap \ballot{i}|}{|\ballot{i}|} \geq \alpha.\]
Moreover, an outcome $\outcome \subseteq \sigma(\af)$ $\alpha$-represents voter $i \in \voters$ iff
$\rep_i(\outcome)=\max_{\pov\in\outcome}\rep_i(\pov) \geq \alpha$.
\end{definition}

\begin{example}[Example~\ref{example:absaf} continued] \label{example:representation}
	Let $\ballot{i} = \{f_2,p_1,p_2\}$, and  
	consider the preferred extensions (points of view) 
		$\pov_1 = \{f_1,f_2,m_1\}$, 
		$\pov_2 = \{p_1,p_2, \allowbreak p_3,m_1\}$, and 
		$\pov_3 = \{f_2, p_1,p_2, \allowbreak m_1\}$. 
	Then 
		$\rep_i(\pov_1) = \nicefrac{1}{3}$, 
		$\rep_i(\pov_2) = \nicefrac{2}{3}$, and 
		$\rep_i(\pov_3) = 1$.
	Regarding outcomes of size~2 we have 
		$\rep_i(\{\pov_1,\pov_2\}) = \nicefrac{2}{3}$ and
		$\rep_i(\{\pov_1,\pov_3\}) = \rep_i(\{\pov_2,\pov_3\}) = 1$.
\end{example}

In deliberative democracy, we usually assume that voters are 
rational agents \citep{Cohen2007} and thus,
in principle, would arrive at a consistent and defendable point of view 
when considering all arguments for a sufficient amount of time, which we could call 
their \emph{ideal} point of view.
However, in practice, the vast majority of participants in an online discussion
will not carefully consider all comments. Hence, we do not assume that
ballots are complete or consistent.
In light of this, we interpret $\rep_i(\pi)$ as a measure of how consistent
the observed voting behavior of $i$ is with
the assumption that $\pi$ is her ideal point of view.
Due to our rationality assumption, ideal points of view should at least be admissible.
As we aim to represent voters with a small number of extensions, we can focus on preferred extensions, i.e., $\sigma = \pref$, since this gives us subset-maximal coherent viewpoints. If an admissible viewpoint $\pi$ contains all approved arguments of an agent, so does the preferred viewpoint $\pi' \supseteq \pi$.
Finally, as two different preferred extensions must by definition be conflicting,
each voter's ideal point of view can only coincide with one preferred extension.
Thus, we define $\rep_i(\outcome)$ as $\max_{\pov\in\outcome}(\rep_i(\pov))$.

\section{Perfect Representation} \label{sec:perfect-representation}

Ideally, we want to find an outcome that perfectly represents everyone, i.e., an outcome that 1-represents all voters. 
This is not always possible in practice, however. Indeed, Example~\ref{example:absaf} clearly demonstrates that real votes cannot be assumed to be conflict-free, in which case there can be no outcome that 1-represents everyone.
%
%
But even if ballots are conflict-free, it may not be possible to fully represent all voters:
consider the ABSAF $(\af,\voters,\ballots)$ from Figure~\ref{fig:example-approving-undefendable-arguments}. The labels above the arguments represent the voters approving of them, i.e., $\ballot{1} = \{a,b\}$, $\ballot{2} = \{a,c,d,e\}$.
Note that $\pref(\af) = \{\{a,b\},\{a,c\}\}$, but $\rep_2(\{a,b\}) = 0.25$, and $\rep_2(\{a,c\}) = 0.5$.


Additionally, deciding whether all voters can be perfectly represented is $\NP$-complete.
This follows from the fact that deciding credulous acceptance\footnote{An argument $x$ is credulously accepted (w.r.t.\ $\sigma$) in an AF $\af$ iff $\exists S \in \sigma(\af) \colon x \in S$.} is $\NP$-complete for preferred semantics~\citep{DvorakDunne2018}.

\begin{definition} \label{def:decision-problem-representability}
  \textsc{$1$-Representability} is the following decision problem: given an ABSAF $\absaf=(\af,\voters,\ballots)$ and $k \in \{1,\ldots,|\voters|\}$, is there $\outcome \subseteq \pref(\af)$ with $|\outcome| \leq k$ such that $\rep_i(\outcome) = 1$ for all $i \in N$?
\end{definition}

\begin{restatable}{proposition}{representationNPhardness} \label{prop:representation-np-hardness}
  \textsc{$1$-Representability} is $\NP$-complete. $\NP$-hardness holds even if there is only one voter, i.e., for $n = 1$. 
\end{restatable}
\begin{proof}
	$\NP$-membership: given $\absaf = (\af,\voters,\ballots)$, guess a set of viewpoints $C = \{\pov_j \subseteq \args \mid 1 \leq j \leq k\}$. Note that $|C|\leq |\voters|$ and, for $1 \leq j \leq k$, $|\pov_j| \leq |\args|$. Thus, the size of $C$ is polynomial in the size of $\absaf$. Now, for each $i \in \voters$, verify that there is some $\pov_j \in C$ such that $\rep_i(\pov_j) = 1$ (clearly possible in polynomial time) and $\pov_j \in \adm(\af)$ (also possible in polynomial time~\citep{DvorakDunne2018}). Note that if $\pov_j \in \adm(\af)$ there is $\pov'_j \in \pref(\af)$ such that $\pov_j \subseteq \pov'_j$. Thus, the outcome $\outcome = \{\pov'_j \mid 1 \leq j \leq k\}$ $1$-represents all $i \in N$.
	
	$\NP$-hardness: let $(\af,x)$ be an arbitrary instance of credulous acceptance for preferred semantics. We construct the ABSAF $\absaf = (\af,\voters,\ballots)$ with a single agent $\voters = \{1\}$ and $\ballot{1} = \{x\}$. Then $x$ is credulously accepted for $\pref$ in $\af$ iff there is $\pov \in \pref(\af)$ with $x \in \pov$ iff there is $\pov \in \pref(\af)$ with $\ballot{1} \subseteq \pov$ iff there is an outcome $\outcome \subseteq \pref(\af)$ with $|\outcome| \leq 1$ that 1-represents voter~1.
\end{proof}

The above result also suggests that finding a straightforward characterization for representability is not possible in the general case. 
However, we can restrict approval-ballots and AFs to guarantee representation.
Specifically, assuming conflict-free votes, we can give a lower-bound for $\alpha$-representation by differentiating between arguments that defend themselves and arguments that do not.

\begin{restatable}{proposition}{conditionsRepresentationPrefLowerbound} \label{prop:conditions-representation-pref-lowerbound}
  Let $\absaf = (\af = (\args,\att),\voters,\ballots)$ be an ABSAF. 
  Let $\mathit{SD}(i) = \{a \! \in \! \ballot{i}  \! \mid \!  (b,a) \! \in \! \att  \Longrightarrow  (a,b) \! \in \! \att\}$ be the self-defending arguments approved by~$i \! \in \! N$.
  If $\ballot{i} \! \in \! \cf(\af)$ 
  and $\frac{|\mathit{SD}(i)|}{|\ballot{i}|} \geq \alpha$
  for all $i  \in N$, 
  then there is $\outcome \subseteq \pref(\af)$ such that $\rep_i(\outcome) \geq \alpha$ for all $i \in N$. 
\end{restatable}
\begin{proof}
	Let $i \in \voters$ be an agent. Note that $\mathit{SD}(i) \in \adm(\af)$ since $\ballot{i} \in \cf(\af)$ by assumption and since each argument in $\mathit{SD}(i)$ defends itself. Thus, there is $\pov_i \in \pref(\af)$ such that $\pov_i \supseteq \mathit{SD}(i)$. Observe that $\rep_i(\pov_i) \geq \frac{|\mathit{SD}(i)|}{|\ballot{i}|} \geq \alpha$. Thus, $\outcome = \{\pov_i \mid i \in \voters\}$ $\alpha$-represents every voter in $\absaf$. 
\end{proof}


If we restrict ourselves to symmetric AFs~\citep{Coste-MarquisDM2005,Dunne2007}, a natural subclass of AFs where $(a,b) \in \att$ implies $(b,a) \in \att$, every agent always approves self-defending arguments. 
Thus, by Proposition~\ref{prop:conditions-representation-pref-lowerbound}, if we are given conflict-free ballots and a symmetric AF $\af$, there is an outcome $\outcome \subseteq \pref(\af)$ such that $\rep_i(\outcome) = 1$ for all $i \in N$. 

So far, we attempted to represent \emph{all} approved arguments for each voter.
However, even if we assume votes to be conflict-free,
some voters might approve of arguments that cannot occur together in a preferred extension, or even arguments that cannot be defended altogether. 
Voter~2 in Figure~\ref{fig:example-approving-undefendable-arguments}, for instance, approves two undefendable
arguments ($d$ and $e$). Thus, $\rep_2(\{a,c\}) = 0.5$ even though $\{a,c\}$ contains \emph{all} arguments in $\ballot{2}$ that can \emph{actually} occur in a viewpoint.
Hence, one could argue that voter~2 is already represented as well as possible.
We therefore introduce an alternative notion of representation that we call \emph{core-representation}.

%

\begin{definition} \label{def:core-representation}
  Let $\absaf=(\af,\voters,\ballots)$ be an ABSAF. 
  For $i \in \voters$ we let
    $\maxDef{i} = \max_{\pov \in \sigma(\af)} |\pov \cap \ballot{i}|$.
  If $\maxDef{i} = 0$ we let $\repDef_i(\pov) = 1$, otherwise
    \[\repDef_i(\pov)=\frac{|\pov \cap \ballot{i}|}{\maxDef{i}}.\]
  We say that $\pov \in \sigma(\af)$ $\alpha$-core-represents $i$ iff $\repDef_i(\pov) \geq \alpha$.
  Moreover, $\outcome \subseteq \sigma(\af)$ $\alpha$-core-represents $i$ iff
  $\repDef_i(\outcome)=\max_{\pov\in\outcome}\repDef_i(\pov) \geq \alpha.$
\end{definition}

  Using core-representation for the ABSAF from Figure~\ref{fig:example-approving-undefendable-arguments}, for voter~1 we get 
$\repDef_1(\{a,b\}) = 1$ and $\repDef_1(\{a,c\}) = 0.5$ 
while for voter~$2$ we get $\repDef_2(\{a,b\}) = 0.5$  and $\repDef_2(\{a,c\}) = 1$. 

In contrast to regular representation, it is always possible to find an outcome that perfectly core-represents every voter.

\begin{observation}
  Given an ABSAF $\absaf=(\af,\voters,\ballots)$, the outcome $\outcome = \pref(\af)$ 1-core-represents every voter $i \in \voters$. 
  Thus, there is an outcome $\outcome' \subseteq \pref(\af)$ of size $|\outcome'| = |\voters|$ that 1-core-represents every $i \in \voters$.
\end{observation}

However, deciding whether there is a \emph{small} outcome that perfectly core-represents every voter is harder than in the case of regular representation,
for which this problem is $\NP$-complete (cf.\ Proposition~\ref{prop:representation-np-hardness}).
We define \textsc{$1$-Core-Representability} analogously to \textsc{$1$-Representability} (cf.\ Definition~\ref{def:decision-problem-representability}), except that we ask for an outcome of size $|\outcome| \leq k$ that $1$-core-represents every voter $i \in N$. 


\begin{restatable}{theorem}{coreRepresentabilityThetaCompleteness} \label{thm:core-representability-completeness}
  \textsc{1-Core-Representability} is $\ThetaP{2}$-complete.
  Moreover, $\ThetaP{2}$-hardness holds even if there are only two voters, i.e., for $n = 2$. 
\end{restatable}

\begin{proof}[Proof of Theorem~\ref{thm:core-representability-completeness} (Membership)]
	Let $\absaf = (\af = (\args,\att), \voters, \ballots)$ be the input ABSAF and $k \in \{1,\ldots,|\voters|\}$.
	We make use of the fact that $\ThetaP{2}$ coincides with $\DeltaParallel{\ell}$~\citep{BussHay1991}, which is the class of problems solvable in polynomial time with $\ell$ rounds of parallel $\NP$-oracle calls. Note that $\ell$ must be a constant.
	For our purposes, $\ell = 2$ rounds suffice. 
	
	In the first round, for every $i \in \voters$ and every $m \in \{1,\ldots,|\args|\}$, use an $\NP$-oracle to decide whether there is $\pov \in \adm(\af)$ such that $|\pov \cap \ballot{i}| = m$. This problem is in $\NP$, since we can guess $\pov \subseteq \args$ and check whether $\pov \in \adm(\af)$ and $|\pov \cap \ballot{i}| = m$. 
	Then, for each voter $i \in \voters$ we compute $\maxDef{i} = \max_{\pov \in \pref(\af)} |\pov \cap \ballot{i}| = \max_{\pov \in \adm(\af)} |\pov \cap \ballot{i}|$.
	
	In the second round, we use a single $\NP$-oracle call to execute the following decision-procedure: guess argument sets $\pov_1,\ldots,\pov_k \subseteq \args$. 
	Then, check that
	\begin{itemize}
		\item $\pov_j \in \adm(\af)$ for every $1 \leq j \leq k$,
		\item $|\pov_j \cap \ballot{i}| = \maxDef{i}$ for every $1 \leq i \leq |\voters|$ and some $1 \leq j \leq k$. 
	\end{itemize}
	Since $\pov_j \in \adm(\af)$ there is a preferred extension $\pov'_j \in \pref(\af)$ such that $\pov_j \subseteq \pov'_j$. For every $1 \leq i \leq |\voters|$ and some $1 \leq j \leq k$ we have $|\pov_j \cap \ballot{i}| \leq |\pov'_j \cap \ballot{i}| \leq \maxDef{i} = |\pov_j \cap \ballot{i}|$, i.e., $|\pov'_j \cap \ballot{i}| = \maxDef{i}$ and therefore $\repDef_i(\pov'_j) = 1$. 
	Thus, for $\Omega = \{\pov'_1,\ldots,\pov'_k\}$ we have $\repDef_i(\outcome) = 1$ for every $i \in N$.
\end{proof}

\begin{proof}[Proof of Theorem~\ref{thm:core-representability-completeness} (Hardness)]
	Let $(\varphi,y)$ be an instance of the $\ThetaP{2}$-hard problem \textsc{CardMaxSat}~\citep{CreignouPW2018}, where we are given a propositional formula $\varphi$ in CNF and a variable $y$, and aks whether there is a cardinality-maximal model of $\varphi$ that contains $y$.
	Let $X$ denote the set of variables contained in $\varphi$, and $C$ the set of clauses. 
	Note that $\varphi$ can be assumed to be satisfiable, since $\varphi$ can be replaced by $\varphi' = \varphi \lor (\bigwedge_{x \in X} (\neg x))$, which can further be transformed into CNF if needed. Indeed, $\varphi'$ is always satisfied by $I = \emptyset$ and $(\varphi,y)$ is a yes-instance of \textsc{CardMaxSat} iff $(\varphi',y)$ is a yes-instance of \textsc{CardMaxSat}.
	
	We construct an instance $(\absaf, k)$ with $k = 1$ and the ABSAF $\absaf = (\af = (\args,\att), \voters, \ballots)$ defined as follows:
	\begin{itemize}
		\item $\args = X \cup \{\overline{x} \! \mid \! x \! \in \! X\} \cup C \cup \varphi^*$, where $\varphi^* = \{\varphi_i \! \mid \! 1 \! \leq \! i \leq \!  |X|  +  1\}$;
		\item $\att = \{ (x,\overline{x}),(\overline{x},x) \mid x \in X \} \cup \{ (x,c) \mid x \! \in \! X, c \! \in \! C, x \! \in \! c \} \cup \{ (\overline{x},c) \mid x \! \in \! X, c \! \in \! C, \neg x \! \in \! c \} \cup \{ (c,c),(c,\varphi_i) \mid c \! \in \! C, \varphi_i \! \in \varphi^* \! \}$;
		\item $\voters = \{1,2\}$;
		\item $\ballots$ is given by $\ballot{1} = X \cup \varphi^*$ and $\ballot{2} = \{y\}$.
	\end{itemize}
	Figure~\ref{fig:reduction-cardmaxsat-corerepresentability} shows the AF constructed from the formula containing the following clauses: $c_1 = (x_1 \lor \neg x_2 \lor x_3)$, $c_2 = (\neg x_1 \lor \neg x_3 \lor x_4)$, $c_3 = (x_2 \lor x_3 \lor \neg x_4)$. If $y = x_1$, then the approval ballots are $\ballot{1} = \{x_1,\ldots,x_4,\varphi_1,\ldots,\varphi_5\}$ and $\ballot{2} = \{x_1\}$.
	
	We show that $(\varphi,y)$ is a yes-instance of \textsc{CardMaxSat} iff there is an outcome $\outcome$ of size $|\outcome| = 1$ that 1-core-represents all $i \in N$.
	
	``$\implies$'': Assume $y$ is true in a cardinality-maximal model $I \subseteq X$ of $\varphi$. 
	Let $\pov = I \cup \{\overline{x} \mid x \in X \setminus I\} \cup \varphi^*$.
	Since $I$ satisfies every clause in $\varphi$, every clause-argument $c \in C$ is attacked by some variable-argument $v \in I \cup \{\overline{x} \mid x \in X \setminus I\}$. Thus, every argument in $\args$ is either in $\pov$ or attacked by $\pov$, i.e.,
	$\pov$ defends itself against all outside attackers. Since, by construction, also $\pov \in \cf(\af)$ we have $\pov \in \adm(\af)$.
	Moreover, $\pov$ is a subset-maximal admissible set, since it is in conflict with all arguments outside $\pov$. Thus, $\pov \in \pref(\af)$. 
	Consider any $\pov' \in \pref(\af)$. There are two cases:
	\begin{itemize}
		\item $\pov' \cap \varphi^* = \emptyset$. Then $|\pov' \cap \ballot{1}| \leq |X| < |\varphi^*| \leq |\pov \cap \ballot{1}|$.
		\item $\pov' \cap \varphi^* \neq \emptyset$. Then $\pov'$ defends some $\varphi_i \in \varphi^*$ against $(c,\varphi_i)$ for every $c \in C$. Thus, $I' = \pov' \cap X$ satisfies all clauses in $\varphi$, i.e., $I' \models \varphi$. Since $I$ is a cardinality-maximal model of $\varphi$, $|I'| \leq |I|$. Therefore, $|\pov' \cap \ballot{1}| \leq |I' \cup \varphi^*| \leq |I \cup \varphi^*| = |\pov \cap \ballot{1}|$.
	\end{itemize}
	We conclude that $\maxDef{1} = |\pov \cap \ballot{1}|$, i.e., $\pov$ 1-core-represents voter~1.
	For voter~2, since $y \in I$ we have $y \in \pov$ and therefore $|\pov \cap(A_2)| = 1$. Observe that $\maxDef{2} = 1$. Thus, $\pov$ also 1-core-represents voter~2.
	
	``$\impliedby$'': Assume there is an outcome $\outcome = \{\pov\} \subseteq \pref(\af)$ of size~1 that 1-core-represents all voters $i \in \voters = \{1,2\}$. 
	Let $I = \pov \cap X$.
	
	Let $I' \subseteq X$ be an arbitrary model of $\varphi$, i.e., $I' \models \varphi$. Such $I'$ exists since $\varphi$ is satisfiable. Let $\pov' = I' \cup \{\overline{x} \mid x \in X \setminus I'\} \cup \varphi^*$. By the same argument as above, $\pov' \in \pref(\af)$. 
	Since $\pov$ 1-core-represents voter~$1$, it must be that $|X|+1 = |\varphi^*| \leq |I' \cup \varphi^*| = |\pov' \cap \ballot{1}| \leq |\pov \cap \ballot{1}| \leq |I \cup \varphi^*|$. By $|X|+1 \leq |\pov \cap \ballot{1}|$ we have $\varphi_i \in \pov$ for at least one $\varphi_i \in \varphi^*$.
	Therefore, $\pov$ defends $\varphi_i$ from all clause-arguments $c \in C$, which means that $I$ satisfies all clauses in $\varphi$ and we have $I \models \varphi$.
	Furthermore, by $|I' \cup \varphi^*| \leq |I \cup \varphi^*|$ we have $|I'| \leq |I|$. Since $I'$ was chosen as an arbitrary model of $\varphi$, we can conclude that $I$ is a cardinality-maximal model of $\varphi$.
	
	Lastly, since $\pov$ 1-core-represents voter~2, and since $\maxDef{2} = 1$, it must be that $y \in \pov$ and therefore $y \in I$.
\end{proof}

\begin{figure}[t]
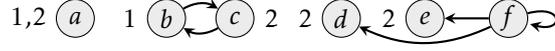

	\centering
	\tikz{
		\node[arg,label={left}:\text{1,2}] (a) at (0,0) {$a$};
		\node[arg,label={left}:\text{1}] (b) at (1.2,0) {$b$};
		\node[arg,label={right}:\text{2}] (c) at (2.1,0) {$c$};
		\node[arg,label={left}:\text{2}] (d) at (3.5,0) {$d$};
		\node[arg,label={left}:\text{2}] (e) at (4.6,0) {$e$};
		\node[arg] (f) at (5.7,0) {$f$};
		\draw[attack] 
		(b) edge [bend left] (c) 
		(c) edge [bend left] (b)
		(f) edge [bend left=25] (d)
		(f) edge (e)
		(f) edge [loop right] (f);
	}
	\caption{Voter~2 approves undefendable arguments.}
	\label{fig:example-approving-undefendable-arguments}
\end{figure}

\begin{figure}[t]
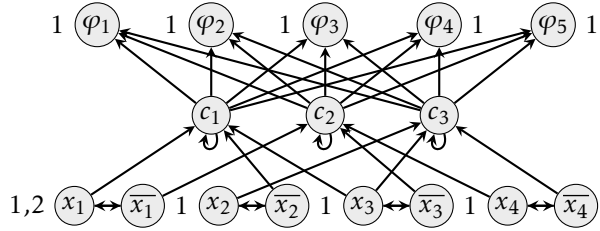

  \centering
  \tikz{
    \node[arg,label={left}:\text{1,2}] (x1) at (0,0) {$x_1$};
    \node[arg] (x1neg) at (0.9,0) {$\overline{x_1}$};
    \node[arg,label={left}:\text{1}] (x2) at (1.9,0) {$x_2$};
    \node[arg] (x2neg) at (2.8,0) {$\overline{x_2}$};
    \node[arg,label={left}:\text{1}] (x3) at (3.8,0) {$x_3$};
    \node[arg] (x3neg) at (4.7,0) {$\overline{x_3}$};
    \node[arg,label={left}:\text{1}] (x4) at (5.7,0) {$x_4$};
    \node[arg] (x4neg) at (6.6,0) {$\overline{x_4}$};
    \node[arg] (c1) at (1.8,1.2) {$c_1$};
    \node[arg] (c2) at (3.3,1.2) {$c_2$};
    \node[arg] (c3) at (4.8,1.2) {$c_3$};
    \node[arg,label={left}:\text{1}] (phi1) at (0.3,2.4) {$\varphi_1$};
    \node[arg,label={left}:\text{1}] (phi2) at (1.8,2.4) {$\varphi_2$};
    \node[arg,label={left}:\text{1}] (phi3) at (3.3,2.4) {$\varphi_3$};
    \node[arg,label={right}:\text{1}] (phi4) at (4.8,2.4) {$\varphi_4$};
    \node[arg,label={right}:\text{1}] (phi5) at (6.3,2.4) {$\varphi_5$};
    \draw[attack] 
    (x1) edge (x1neg)
    (x1neg) edge (x1)
    (x2) edge (x2neg)
    (x2neg) edge (x2)
    (x3) edge (x3neg)
    (x3neg) edge (x3)
    (x4) edge (x4neg)
    (x4neg) edge (x4)
    (x1) edge (c1)
    (x3) edge (c1)
    (x2neg) edge (c1)
    (x3neg) edge (c2)
    (x1neg) edge (c2)
    (x4) edge (c2)
    (x2) edge (c3)
    (x3) edge (c3)
    (x4neg) edge (c3)
    (c1) edge[loop below, looseness=5] (c1)
    (c2) edge[loop below, looseness=5] (c2)
    (c3) edge[loop below, looseness=5] (c3)
    (c1) edge (phi1)
    (c1) edge (phi2)
    (c1) edge (phi3)
    (c1) edge (phi4)
    (c1) edge (phi5)
    (c2) edge (phi1)
    (c2) edge (phi2)
    (c2) edge (phi3)
    (c2) edge (phi4)
    (c2) edge (phi5)
    (c3) edge (phi1)
    (c3) edge (phi2)
    (c3) edge (phi3)
    (c3) edge (phi4)
    (c3) edge (phi5);
  }
  \caption{Construction used in the proof of Theorem~\ref{thm:core-representability-completeness}.}
  \label{fig:reduction-cardmaxsat-corerepresentability}
\end{figure}

If the number of arguments/extensions in our given ABSAF is small, however, we can deal with representability more efficiently.

\begin{restatable}{proposition}{coreRepresentabilityFPT} \label{prop:core-representation-is-fpt}
  \textsc{1-Representability} and \textsc{1-Core-Representability} are FPT with respect to the number of arguments in the given ABSAF. 
\end{restatable}
\begin{proof}
  Let $\absaf = (\af = (\args,\att), \voters, \ballots)$ and $k \in \{1,\ldots,|\voters|\}$. Let $m = |\pref(\af)|$.
  Note that $m$ is in $O^*(3^{|\args|/3}) \subseteq O(2^{|\args|})$~\citep{DunneDLW15}, and that $\pref(\af)$ can be enumerated in $O^*(3^{|\args|/3})$ time~\citep{GaspersL19}.
  Moreover, we assume $k \leq m$ since $\outcome \subseteq \pref(\af)$ for any outcome $\outcome$.
  
  For core-representation, determine $\maxDef{i}$ for all $i \in \voters$ by enumerating every $\pov \in \pref(\af)$ and computing $|\pov \cap \ballot{i}|$. $\maxDef{i}$ is the size of the largest intersection. This runs in $O^*(3^{|\args|/3})$ time.
  
  Then, check if there is a $k$-tuple $\outcome$ of preferred extensions such that  $\rep_i(\outcome) = 1$ (resp.\ $\repDef_i(\outcome) = 1$) for every $i \in \voters$. There are ${m \choose k}$ such $k$-tuples, i.e., this can be done in $O^*({m \choose k}) \subseteq O^*(m^k)$ time. 
\end{proof}

The results in the remainder of the paper apply to both regular representation (Definition~\ref{def:representation}) and core-representation (Definition~\ref{def:core-representation}). 

We have seen that we generally cannot $1$-represent all voters,
but we can $1$-core-represent them. If the AF is small enough,
we can even do so efficiently. On the other hand, when dealing with larger instances,
we may require a huge number of preferred extensions to $1$-(core-)represent
all voters which does not help to understand the discussion.
In this case, we might instead aim for an optimal representation with a fixed number of preferred extensions.

\section{Optimizing Representation}




Let us now consider the question how to optimally represent the voters in an ABSAF with
a fixed number of preferred extensions.
One idea, commonly used in social choice theory and referred to as the \emph{Utilitarian rule}, is to pick an outcome $\outcome$ maximizing the average representation across all voters\footnote{We assume that an arbitrary tiebreaking mechanism is used for all rules.}, i.e., $\sum_{i\in \voters}\rep_i(\outcome)$. The second standard approach, referred to as the \emph{Egalitarian rule}, is to pick an outcome maximizing the representation of the least-represented voter, i.e., $\min_{i\in \voters}\rep_i(\outcome)$. 

A family of rules generalizing these ideas is based on \emph{ordered weighted averaging} (OWA) vectors~\citep{yager1988ordered}.
Given an outcome $\outcome$, let $\vec{s}(\outcome)=(s_1,\ldots,s_n)$ be the vector $(\rep_1(\outcome),\ldots,\rep_n(\outcome))$ sorted in non-decreasing order (i.e., $s_1$ is the least-represented voter). For a non-increasing vector of non-negative weights $\vec{\weight}=(\weight_1,\ldots,\weight_n)$, where $\weight_1>0$, the corresponding OWA rule is defined as:
\[ \OWA(\absaf) \in \argmax_{\outcome\subseteq\pref(\af)\colon|\outcome|\leq k} \vec{\weight}\cdot\vec{s}(\outcome) . \]
Here, $\cdot$ is the dot product. The Utilitarian and Egalitarian rules are given by $\vec{\weight}=(1,\ldots,1)$ and $\vec{\weight}=(1,0,\ldots,0)$, respectively. We will also consider the \emph{Harmonic rule} based on the vector $(1,\nicefrac{1}{2},\ldots,\nicefrac{1}{n})$. This sequence of weights is often used to achieve proportionality in multiwinner voting \citep{abcbook}.
OWA-rules in general have been studied, e.g., in the context of multiple referenda/issues~\citep{amanatidis2015multiple, barrot2017manipulation, LacknerEtAlAAMAS2023}.

An alternative approach is simply to maximize the number of voters that are $1$-represented. We call this rule $\maxCov$, defined as:
\[ \maxCov(\absaf) \in \argmax_{\outcome\subseteq\pref(\af)\colon|\outcome|\leq k} |\{i\in \voters\colon \rep_i(\outcome)=1 \}| . \]

\begin{example}
	Consider an ABSAF with voters $\voters = \{1,2,3,4\}$ and extensions $\pref(\af) = \{\pov_1,\pov_2,\pov_3\}$. 
	Say that we want to pick an outcome of size $k = 1$ 
	and that the representation scores are as follows:
	\begin{center}
	\begin{tabular}{c|cccc}
\toprule
		$\outcome$ & $\rep_1(\outcome)$ & $\rep_2(\outcome)$ & $\rep_3(\outcome)$ & $\rep_4(\outcome)$ \\ \midrule
		$\{\pov_1\}$ & 0 & 1 & 1 & 1 \\
		$\{\pov_2\}$ & 0.4 & 0.5 & 1 & 1 \\
		$\{\pov_3\}$ & 0.5 & 0.5 & 0.5 & 0.5\\
\bottomrule
	\end{tabular}
	\end{center}
	It can be verified that Utilitarian and \maxCov\ select $\{\pov_1\}$, Harmonic selects $\{\pov_2\}$, and Egalitarian selects $\{\pov_3\}$.
\end{example}


\begin{example}
	Selecting outcomes of size $k=2$
	for the Canada electoral reform data (cf.\ Example~\ref{example:absaf}),
	Utilitarian and \maxCov\ pick
	$\outcome_1 = \{\{p_1, p_2, \allowbreak p_3, \allowbreak s_1, s_2\}, \allowbreak \{f_2, p_1, \allowbreak p_2,m_1\}\}$ 
	while Harmonic and Egalitarian pick
	$\outcome_2 = \{\{p_1, p_2, \allowbreak p_3, \allowbreak s_1, s_2\}, \allowbreak  \{f_1,f_2, \allowbreak m_1\}\}$.
	The minimum and average (core-)representation-scores 
	across all voters are:
	\begin{center}
	\begin{tabular}{c|cccc}
\toprule
		& min & avg & min (core) & avg (core)\\ \midrule
		$\outcome_1$ & 0 & 0.9378 & 0 & 0.9706 \\
		$\outcome_2$ & 0.5 & 0.9360 & 0.5 & 0.9697  \\
\bottomrule
	\end{tabular}
	\end{center}	
  	The average representation for Utilitarian is only marginally larger than for Egalitarian/Harmonic in this instance.
\end{example}

Note that all rules achieve perfect representation if possible.

\begin{restatable}{proposition}{allRulesFindPerfectOutcomes}\label{prop:all-rules-find-perfect-outcomes}
  The Egalitarian rule returns an outcome of size~$k$ in which every voter is $\alpha$-represented iff such an outcome exists. 
  Moreover, every OWA rule, as well as \maxCov, returns an outcome of size $k$ in which every voter is 1-represented iff such an outcome exists. 
\end{restatable}
\begin{proof}
	Suppose that there is an outcome $\outcome$ of size $k$ that $\alpha$-represents all voters (for some $\alpha\in[0..1]$). Assume towards a contradiction that Egalitarian chooses an outcome $\outcome^\prime$ that does not $\alpha$-represent all voters. Thus, the least-satisfied voter will have satisfaction $s<\alpha$. But with outcome $\outcome$ the least-satisfied voter would have higher satisfaction (at least $\alpha$). Contradiction: Egalitarian must choose an outcome that $\alpha$-represents all voters.
	
	Next, suppose that there is an outcome $\outcome$ of size $k$ that $1$-represents all voters. It follows from definition that \maxCov{} must select this outcome or any other outcome that $1$-represents all voters. Now, consider a generic rule $\OWA$. Assume towards a contradiction that $\OWA$ chooses an outcome $\outcome^\prime$ that does not $1$-represent all voters. Let $\vec{s}(\outcome^\prime)=(s_1,\ldots,s_n)$. We get the following:
	\begin{align*}
	\vec{s}(\outcome)\cdot\vec{\weight} &> \vec{s}(\outcome^\prime)\cdot\vec{\weight} && \iff\\
	\sum_{i=1}^n \weight_i &> \sum_{i=1}^n s_i\weight_i && \iff\\
	\sum_{i=1}^n(1-&s_i)\weight_i > 0.
	\end{align*}
	
	Recall that $\vec{\weight}$ is non-negative and that, for all $i\in\voters$, $s_i\in[0..1]$. Furthermore, recall that $\weight_1>0$, and notice that, since $\outcome^\prime$ does not $1$-represent all voters, we have $s_1<1$. These facts together show that the above inequality holds. Contradiction: $\outcome$ has a higher OWA-score, meaning that $\OWA$ cannot pick $\outcome^\prime$.
\end{proof}

Moreover, all of the rules considered here are computationally intractable.\footnote{Note that the problem of computing an optimal outcome w.r.t\ a given rule is an optimization problem, not a decision problem. We show that these optimization problems are $\NP$-hard, i.e., that we can not solve them in polynomial time unless $\Polytime = \NP$.} 
This follows from Proposition~\ref{prop:representation-np-hardness} and~\ref{prop:all-rules-find-perfect-outcomes}.

\begin{restatable}{proposition}{rulesAreIntractable} \label{prop:rules-are-intractable}
  Computing an outcome that is optimal with respect to a given OWA-rule or \maxCov\ is $\NP$-hard. 
\end{restatable}
\begin{proof}
	Recall that deciding whether all voters in a given ABSAF can be $1$-represented is $\NP$-hard even if there is only a single voter that approves a single argument, i.e., $\ballot{1} = \{x\}$ (cf.\ proof of Proposition~\ref{prop:representation-np-hardness}). 
	An outcome $\outcome$ either $1$-represents the voter (if $x \in \pov$ for some $\pov \in \outcome$) or $0$-represents the voter (otherwise). 
	By Proposition~\ref{prop:all-rules-find-perfect-outcomes}, all considered rules will return an outcome that 1-represents the voter if such an outcome exists. Thus, the rules cannot be computed in polynomial time unless $\Polytime = \NP$. 
\end{proof}

While intractable, $\OWA(\absaf)$ and $\maxCov$ rules are FPT: 
analogously to Proposition~\ref{prop:core-representation-is-fpt}, we can first enumerate all preferred extensions in $O^*(3^{|\args|/3})$ time and then enumerate all outcomes of size~$k$ in $O^*({m \choose k}) \subseteq O^*(m^k)$ time, where $m = |\pref(\af)|$, allowing us to simply choose the best outcome w.r.t.\ the given rule. 

In practice, however, a runtime of $O^*({m \choose k})$ can be impractical, as we will see in Section~\ref{sec:experiments}. Thus, we introduce greedy variants for all of the above rules as follows. Consider first the greedy variant $\greedOWA$ of a rule $\OWA$. Assuming that we have already selected $\ell$ viewpoints $\pov_1,\ldots,\pov_\ell$, we pick the $(\ell+1)$-th as follows:
\[ \pov_{\ell+1} \in \argmax_{\pov\in\pref(\af)\setminus\{\pov_1,\ldots,\pov_\ell\}} \vec{\weight}\cdot\vec{s}(\{\pov_1,\ldots,\pov_\ell,\pov\}) . \]
We stop as soon as $k$ points of view have been selected. The greedy variant of $\maxCov$, called $\greedCov$, is defined analogously. 
Observe that \greedCov\ approximates \maxCov\ with a factor of $1 - \nicefrac{1}{e}$, which follows directly from the approximation guarantee of the greedy algorithm for the Maximum Coverage Problem~\citep{Hochbaum1997}.

The greedy rules are also intractable, just like their non-greedy variants (cf. Proposition~\ref{prop:rules-are-intractable}). The proof is analogous.

\begin{restatable}{proposition}{greedyRulesAreIntractable} \label{prop:greedy-rules-are-intractable}
	Computing an outcome that is optimal with respect to a given greedy OWA-rule or \greedCov\ is $\NP$-hard. 
\end{restatable}

However, we can improve upon the FPT-algorithm for the non-greedy rules since we do not need to enumerate all outcomes. 
Rather, it suffices to enumerate all points of view whenever we pick a new point of view in the greedy procedure. 
Thus, we can first enumerate all preferred extensions in $O^*(3^{|\args|/3})$ time and then execute the greedy procedure in $O^*(m k)$, where $m = |\pref(\af)|$. 

If we assume that we are given the preferred extensions, e.g., by precomputation via powerful argumentation solvers~\citep{egly2008aspartix,NiskanenJ20a}, we have a runtime of $O^*(m k)$ for the greedy rules.
This is polynomial in $m$ since an outcome can contain at most $m$ viewpoints, i.e., $k \leq m$. 
Assuming $\Polytime \neq \NP$, such an FPT-algorithm does not exist for the main non-greedy OWA-rules.

\begin{restatable}{proposition}{greedyIsFasterThanExactMethods} \label{prop:greedy-normal-complexity}
  Assuming $\Polytime \neq \NP$, there is no algorithm that, given an ABSAF $\absaf = (\af,\voters,\ballots)$ and given $\pref(\af)$, returns an outcome for $\absaf$ that is optimal with respect to the Utilitarian, Egalitarian, or Harmonic rule in polynomial time with respect to $m = |\pref(F)|$.
\end{restatable}
\begin{proof}
	For Utilitarian and Harmonic we consider approval based committee voting under the Chamberlin–Courant (CC) rule~\citep{ChamberlinCourant1983,abcbook}, where we are given a set of voters $N = \{1,\ldots,n\}$, a set $C$ of $m = |C|$ candidates, a set $B(i) \subseteq C$ of approved candidates for each $i \in N$, and a number $k$. The task is to find a subset $W \subseteq C$ of size $|W| = k$ that maximizes $|\{i \in N \colon B(i) \cap W \neq \emptyset\}|$. Computing an optimal $W \subseteq C$ under CC is $\NP$-hard~\citep{ProcacciaRZ2008}.  
	We construct the ABSAF $\absaf = (\af = (\args,\att),\voters,\ballots)$ with 
	\begin{itemize}
		\item $\args = C \cup C'$, where $C' = \{c' \mid c \in C\}$;
		\item $\att = \{(c,d) \mid c,d \in C \text{ and } c \neq d\} \cup \{(c',c') \mid c' \in C'\}$;
		\item $\ballot{i} = B(i) \cup \{c' \mid c \in C \setminus B(i)\}$ for every $i \in N$.
	\end{itemize}
	Note that $\pref(\af) = \{\{c\} \mid c \in C\}$ and that $|\pref(\af)| = m$. Moreover, $|\ballot{i}| = m$ for every $i \in N$. Thus, for every $\outcome \subseteq \pref(\af)$, if there is $\{c\} \in \outcome$ such that $c \in \ballot{i}$, we have $\rep_i(\outcome) = \nicefrac{1}{m}$ and $\repDef_i(\outcome) = 1$. Otherwise, $\rep_i(\outcome) = \repDef_i(\outcome) = 0$. 
	Hence, $W \subseteq C$ maximizes $|\{i \in N \colon B(i) \cap W \neq \emptyset\}|$ iff $\outcome = \{\{c\} \mid c \in W\}$ is optimal with respect to Utilitarian/Harmonic. Now, if there was an algorithm that could compute $\outcome$ in polynomial time with respect to $m$, then we could compute $W$ in polynomial time, which would imply $\Polytime = \NP$.  
	
	For Egalitarian, let $(C,S_1,\ldots,S_\ell,k)$ be an instance of the prototypical $\NP$-complete Hitting Set problem~\citep{Karp1972}. Specifically, $C$ is a set of $m$ elements, $S_j \subseteq C$ for every $1 \leq j \leq \ell$, and $k$ is a natural number. $H \subseteq C$ is a hitting set iff $H \cap S_j \neq \emptyset$ for every $1 \leq j \leq \ell$. We must decide whether there is a hitting set $H$ of size $|H| \leq k$. 
	
	We construct the ABSAF $\absaf = (\af = (\args,\att),\voters,\ballots)$ with 
	\begin{itemize}
		\item $\args = C \cup C'$, where $C' = \{c' \mid c \in C\}$;
		\item $\att = \{(c,d) \mid c,d \in C \text{ and } c \neq d\} \cup \{(c',c') \mid c' \in C'\}$;
		\item $\ballot{i} = S_i \cup \{c' \mid c \in C \setminus S_i\}$ for every $i \in N$.
	\end{itemize}
	As above, $\pref(\af) = \{\{c\} \mid c \in C\}$. Therefore, for every $\outcome \subseteq \pref(\af)$, if there is $\{c\} \in \outcome$ such that $c \in \ballot{i}$ we have $\rep_i(\outcome) = \nicefrac{1}{m}$ and $\repDef_i(\outcome) = 1$. Otherwise, $\rep_i(\outcome) = \repDef_i(\outcome) = 0$. 
	
	Now $H \subseteq C$ is a hitting set iff $\min_{i \in \voters}\rep_i(\outcome) =  \nicefrac{1}{m}$ (resp.\ $\min_{i \in \voters}\repDef_i(\outcome) =  1$) for $\outcome = \{\{c\} \mid c \in H\}$. 
	Thus, if there was an algorithm that could compute an optimal outcome for Egalitarian in polynomial time with respect to $m$, then we could compute a minimal hitting set in polynomial time, which would imply $\Polytime = \NP$.
\end{proof}

We have now introduced several voting rules. Next, we will compare them axiomatically
and in numerical simulations.

\section{Justified Representation}

One of the key axioms of multiwinner voting is 
justified representation \citep{aziz2017justified}, which requires that each sufficiently large group
of voters sharing the same opinion is represented in the outcome. This is
also a natural desideratum in our setting. To formalize it, we first need 
to define what it means for a group to have the same opinion.

\begin{definition} \label{def:representable-groups}
	Let $\absaf=(\af,\voters,\ballots)$ be an ABSAF with $n = |\voters|$ voters. 
	We call a group of voters $\voters' \subseteq \voters$ \emph{$1$-representable} iff there is $\pov \in \pref(\af)$ such that for all $i \in \voters'$ we have $\rep_i(\pov) = 1$.
\end{definition}
 
If such a group contains
more than $\nicefrac{n}{k}$ voters, where $k$ is the number of preferred extensions we want to select, then this group, arguably, deserves to be fully represented. 

\begin{definition} \label{def:sjr}
	An outcome $\outcome \subseteq \pref(\af)$ of size $|\outcome|=k$ satisfies \emph{Strong Justified Representation} (SJR) iff for every $1$-representable group $\voters'\subseteq \voters$ of size $|\voters'| \geq \nicefrac{n}{k}$ there is a point of view $\pov \in \outcome$ such that $\rep_i(\pov) = 1$ for every $i \in \voters'$.
\end{definition}

Observe that Definitions~\ref{def:representable-groups} and~\ref{def:sjr} could be written using the notion of core-representation (see Definition~\ref{def:core-representation}) instead. All results that we present in this section hold in either case. 

\begin{figure}[t]
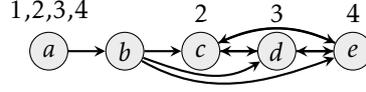

  \centering
  \tikz{
    \node[arg,label={above}:\text{1,2,3,4}] (a) at (0,0) {$a$};
    \node[arg] (b) at (1,0) {$b$};
    \node[arg,label={above}:2] (c) at (2,0) {$c$};
    \node[arg,label={above}:3] (d) at (3,0) {$d$};
    \node[arg,label={above}:4] (e) at (4,0) {$e$};
    \draw[attack] 
    (a) edge (b)
    (b) edge (c)
    (b) edge[out=-20,in=212] (d)
    (b) edge[out=-30,in=200] (e)
    (c) edge (d)
    (c) edge[bend left=25] (e)
    (d) edge (c)
    (d) edge (e)
    (e) edge[bend right=25] (c)
    (e) edge (d);
  }
  \caption{Counterexample SJR.}
  \label{fig:counterexample-SJR}
\end{figure}
\begin{figure}[t]
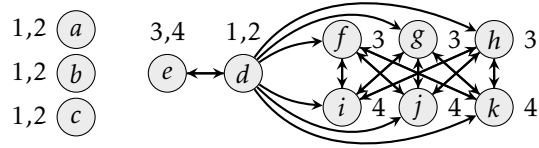

	\centering
	\tikz{
		\node[arg,label={left}:\text{1,2}] (a) at (0,1.58) {$a$};
		\node[arg,label={left}:\text{1,2}] (b) at (0,1) {$b$};
		\node[arg,label={left}:\text{1,2}] (c) at (0,0.42) {$c$};
		\node[arg,label={above}:\text{1,2}] (d) at (2.2,1) {$d$};
		\node[arg,label={above}:\text{3,4}] (e) at (1.2,1) {$e$};
		\node[arg,label={right}:\text{3}] (f) at (3.5,1.45) {$f$};
		\node[arg,label={right}:\text{3}] (g) at (4.5,1.45) {$g$};
		\node[arg,label={right}:\text{3}] (h) at (5.5,1.45) {$h$};
		\node[arg,label={right}:\text{4}] (i) at (3.5,0.55) {$i$};
		\node[arg,label={right}:\text{4}] (j) at (4.5,0.55) {$j$};
		\node[arg,label={right}:\text{4}] (k) at (5.5,0.55) {$k$};
		\draw[attack] 
		(d) edge (e)
		(e) edge (d)
		(d) edge[out=35,in=185] (f)
		(d) edge[out=45, in=150] (g)
		(d) edge[out=55, in=155] (h)
		(d) edge[out=-35,in=-185] (i)
		(d) edge[out=-45, in=-150] (j)
		(d) edge[out=-55, in=-155] (k)
		
		(f) edge (i)
		(i) edge (f)
		(f) edge (j)
		(j) edge (f)
		(f) edge (k)
		(k) edge (f)
		(g) edge (i)
		(i) edge (g)
		(g) edge (j)
		(j) edge (g)
		(g) edge (k)
		(k) edge (g)
		(h) edge (i)
		(i) edge (h)
		(h) edge (j)
		(j) edge (h)
		(h) edge (k)
		(k) edge (h);
	}
	\caption{Counterexample JR (OWA rules).}
	\label{fig:counterexample-OWA-JR}
\end{figure}

As can be seen quite easily,
if the approval sets of all voters are preferred extensions, then SJR can be satisfied.

\begin{restatable}{proposition}{SJRifVotesArePrefExtensions}
Let $\absaf=(\af,\voters,\ballots)$ be an ABSAF. If $\ballot{i} \in \pref(\af)$ for all $i \in \voters$, then for any $k \in \{1,\ldots,|\pref(\af)|\}$ we can find an outcome $\outcome$ of size $|\outcome| = k$ that satisfies SJR.
\end{restatable}
\begin{proof}
	Let $\absaf=(\af,\voters,\ballots)$ be an ABSAF with $n = |\voters|$ voters.
	For every two preferred extensions $\pov,\pov' \in \pref(\af)$ we know that $\pov \not \subset \pov'$ and vice versa. 
	Hence, if the approval set $\ballot{i}$ of a voter $i \in \voters$ is a preferred extension $\pov$, then $i$ is only 1-represented by $\pov$.  
	Now let $N_1$, $N_2$, \dots, $N_\ell$ be all 1-representable groups of size $\nicefrac{n}{k}$. 
	As every voter can only be 1-represented by one extension, these groups must be 
	disjoint. Since additionally we have $|N_j| \geq \nicefrac{n}{k}$
	it follows that $\ell \leq k$ holds and hence we can just select for every $N_j$
	the preferred extension that 1-represents the voters in $N_j$. 
\end{proof}

Unfortunately, once we drop this unreasonably strong assumption, the result does not hold. 

\begin{restatable}{proposition}{SJRcanNotBeGuaranteed}
  It cannot be guaranteed that, given an ABSAF $\absaf=(\af,\voters,\ballots)$ and $k \in \{1,\ldots,|\voters|-1\}$, there is an outcome $\outcome$ of size $|\outcome| = k$ that satisfies SJR (even if votes are assumed to be admissible).
\end{restatable}
\begin{proof}
	Let $\absaf=(\af,\voters,\ballots)$ be the ABSAF shown in Figure~\ref{fig:counterexample-SJR}. There are three preferred extensions, namely $\pref(\af) = \{\{a,c\}, \allowbreak \{a,d\}, \allowbreak \{a,e\}\}$. However, we cannot pick an outcome $\outcome \subseteq \pref(\af)$ of size $k = 2$ that satisfies SJR: Note that the voter groups $\{1,2\}, \allowbreak \{1,3\}, \allowbreak \{1,4\}$ are all $1$-representable. Moreover, these groups have a size of 2 while also $\nicefrac{n}{k} = \nicefrac{4}{2} = 2$.
	However, if we pick the outcome $\outcome_1 = \{\{a,c\},\{a,d\}\}$ we have that $\rep_4(\pov) = 0.5$ despite voter~4 being part of a $1$-representable group of size 2. Analogously if we pick $\outcome_2 = \{\{a,c\},\{a,e\}\}$ or $\outcome_3 = \{\{a,d\},\{a,e\}\}$.
\end{proof}

This motivates us to introduce the following weakening of SJR, where we only require that at least one voter of every sufficiently large $1$-representable group must be 1-represented. Analogous restrictions have also been proposed in the context of multiwinner voting~\citep{abcbook}
and Participatory Budgeting~\citep{rey2023computational}.

\begin{definition} \label{def:jr}
	An outcome $\outcome \subseteq \pref(\af)$ with $|\outcome|=k$ satisfies \emph{Justified Representation} (JR) iff for every $1$-representable group ${\voters'\subseteq \voters}$ with $|\voters'| \geq \nicefrac{n}{k}$ there is $\pov \in \outcome$ such that $\rep_i(\pov) = 1$ for some $i \in \voters'$.
\end{definition}

In order for an outcome $\outcome$ of size $k$ to satisfy JR, the set of voters in $\voters'$ that are not 1-represented by $\outcome$ must be smaller than $\nicefrac{n}{k}$. Indeed, if $\nicefrac{n}{k}$ or more voters of a $1$-representable group $\voters'$ were not 1-represented by $\outcome$, then
$\voters'' = \{i \in \voters' \mid \forall \pov \in \outcome \; \rep_i(\pov) < 1\}$
would be a $1$-representable group of size $\nicefrac{n}{k}$ that violates JR.

Fortunately, in contrast to SJR, JR can always be satisfied. Indeed, in contrast to multiwinner voting
it can be satisfied by just maximizing the number of 1-represented voters.

\begin{restatable}{theorem}{maxCovJR} \label{thm:maxCov-JR}
\maxCov{} and \greedCov{} satisfy JR. 
\end{restatable}

\begin{proof}
	Let $\outcome$ be the outcome of either \maxCov{} and \greedCov{} with size $|\outcome| = k$.
	Assume that there is a $1$-representable group $\voters'$ of size at least $\nicefrac{n}{k}$ that is
	represented by $\pov'$ such that no voter in $\voters'$ is 1-represented in $\outcome$.
	Thus, for the set of 1-represented voters
	\[N_{\outcome}^1:= \{i \in \voters \mid \exists \pov \in \outcome \; \rep_i(\pov) = 1\}\]
	we have $|N_{\outcome}^1| \leq n - |\voters'| \leq n - \nicefrac{n}{k}$.
	Now let $\pov_1, \pov_2, \dots, \pov_k$ be an enumeration of the elements of $\outcome$.
	In the case of \greedCov{} let this be the sequence in which the elements of $\outcome$
	are picked, in the case of \maxCov{} let the enumeration be arbitrary.
	Now we partition $N_{\outcome}^1$ into the following sets 
	\[N_{\pov_j}^1 := \{i \in N_{\outcome}^1 \mid \rep_i(\pov_j) = 1 \land \forall \ell < j \;
	\rep_i(\pov_\ell) < 1\}.\]
	In other words, $N_{\pov_j}^1$ is the set of voters for which $\pov_j$ is the minimal
	extension 1-representing them. As this is a partition we have 
	\[\sum_{j=1}^k |\voters^1_{\pov_j}| = |\voters^1_{\outcome}| \leq n - \frac{n}{k} < n\]
	and therefore there must be $\ell \leq k$ for which
	$|\voters^1_{\pov_\ell}| < \nicefrac{n}{k} \leq |\voters'|$ holds.
	Assume now that $\outcome$ was chosen by \maxCov. Observe that
	\[|N_{\outcome}^1| = \sum_{j=1}^k |\voters^1_{\pov_j}|  <
	\left(\sum_{j=1}^k |\voters^1_{\pov_j}|\right) - |\voters^1_{\pov_\ell}| + |\voters'| \leq
	|N_{\outcome\cup \{\pov'\} \setminus \{\pov_\ell\}}^1|\]
	where the last inequality follows from the fact that no voter in $\voters'$ is represented by 
	$\outcome$. However, this contradicts the assumption that $\outcome$ was chosen by
	\maxCov.
	Assume now that $\outcome$ was chosen by \greedCov. Then, in round $\ell$, we have
	\[|N_{\bigcup_{j=1}^\ell \pov_j}^1| = \sum_{j=1}^\ell |\voters^1_{\pov_j}|  <
	\left( \sum_{j=1}^{\ell-1} |\voters^1_{\pov_j}| \right) + |\voters'| =
	|N_{(\bigcup_{j=1}^{\ell-1} \pov_j) \cup \{\pov'\}}^1|\]
	where the last inequality follows again from the fact that no voter
	in $\voters'$ is represented by $\outcome$. However, this contradicts 
	the assumption that $\pov_\ell$ was chosen by \greedCov{} in round $\ell$.
\end{proof}

In contrast, rules that do not explicitly maximize the number of 1-represented agents do not satisfy JR.


\begin{restatable}{proposition}{OWAcanNotGuaranteeJR}\label{prop:OWA-JR}
  No OWA or greedy OWA rule is guaranteed to return an outcome that satisfies JR, irrespective of the tie-breaking used, and even if votes are admissible.
\end{restatable}
\begin{proof}
	Fix a weight vector $\vec{\weight}$ and let $k=2$. Consider the ABSAF $\absaf=(\af,\voters,\ballots)$ shown in Figure~\ref{fig:counterexample-OWA-JR}. $F$ has preferred extensions
	$\pov_1 = \{a,b,c,d\}$,
	$\pov_2 = \{a,b,c,e,f,g,h\}$, and
	$\pov_3 = \{a,b,c,e,i,j,k\}$.
	There are $n = 4$ voters with approvals 
	$\ballot{1} =\ballot{2} = \pov_1$,
	$\ballot{3} = \{e,f,g,h\}$, and
	$\ballot{4} = \{e,i,j,k\}$.
	Note that there is a $1$-representable group of size at least $\nicefrac{n}{k} = 2$, namely $\voters' = \{1,2\}$. Thus, to satisfy JR, we need to pick an outcome containing $\pov_1$.
	
	We first show that the outcome returned by the $\OWA$ rule does not satisfy JR.
	It can be verified that
	the possible sorted representation vectors are
	$\vec{s}_1 = \vec{s}(\{\pov_1,\pov_2\}) = \vec{s}(\{\pov_1,\pov_3\}) = (\nicefrac{1}{4}, 1, 1, 1)$ and
	$\vec{s}_2 = \vec{s}(\{\pov_2,\pov_3\}) = (\nicefrac{3}{4}, \nicefrac{3}{4}, 1, 1)$.
	Therefore,
	\begin{align*}
	\vec{s}_2\cdot \vec{\weight}&>\vec{s}_1\cdot \vec{\weight} && \iff \\
	\nicefrac{3}{4}\weight_1+\nicefrac{3}{4}\weight_2+\weight_3+\weight_4 & > \nicefrac{1}{4}\weight_1+\weight_2+\weight_3+\weight_4 && \iff \\
	2\weight_1 -\weight_2 & > 0.
	\end{align*}
	The last inequality holds by the fact that $\vec{\weight}$ is a non-increasing vector and $w_1>0$. Thus, outcome $\{\pov_2,\pov_3\}$ is chosen, failing JR.
	
	We now show that $\greedOWA$ also fails JR. In the first round, we have the sorted representation vectors
	$\vec{s}_3 = \vec{s}(\{\pov_1\}) = (0, 0, 1, 1)$ and
	$\vec{s}_4 = \vec{s}(\{\pov_2\}) = \vec{s}(\{\pov_3\}) = (\nicefrac{1}{4}, \nicefrac{3}{4}, \nicefrac{3}{4}, 1)$.
	Consequently,
	\begin{align*}
	\vec{s}_4\cdot \vec{\weight} > \vec{s}_3\cdot \vec{\weight} \iff 
	\weight_1+3\weight_2-\weight_3 > 0.
	\end{align*}
	Again, the last inequality holds by definition of $\vec{\weight}$. Hence, the greedy rule will either pick $\{\pov_2\}$ or $\{\pov_3\}$. Assume w.l.o.g. that $\{\pov_2\}$ is chosen. Then, in the second (and final) round, the rule will only compare outcomes containing $\pov_2$, namely $\{\pov_1,\pov_2\}$ and $\{\pov_2,\pov_3\}$. However, we have shown above that the latter has a higher $\OWA$ score. Thus, $\{\pov_2,\pov_3\}$ is chosen, failing JR.
\end{proof}

Theorem~\ref{thm:maxCov-JR} and Proposition~\ref{prop:OWA-JR} together show that \maxCov{} (resp.\ \greedCov{}) does not belong to the class of OWA rules (resp.\ greedy OWA rules) and
that \maxCov{} outperforms all OWA rules with respect to the key 
axiom of justified representation.

\section{Experiments} \label{sec:experiments}

\begin{figure*}
	\centering
	\begin{subfigure}[b]{0.45\textwidth}
		\centering
		\includegraphics[width=\textwidth]{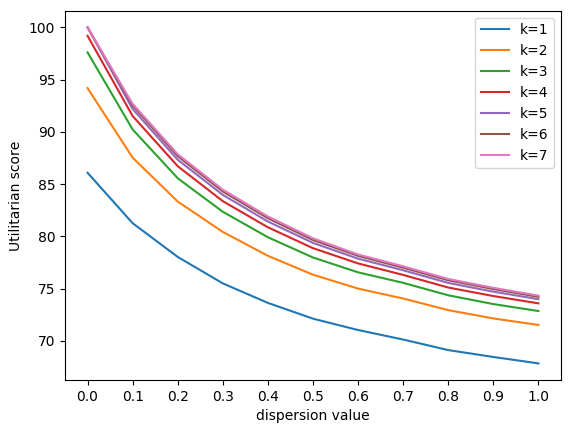}
		\caption{Varying $k$ experiment.}
		\label{fig:varyingExp}
	\end{subfigure}
	\hfill
	\begin{subfigure}[b]{0.45\textwidth}
		\centering
		\includegraphics[width=\textwidth]{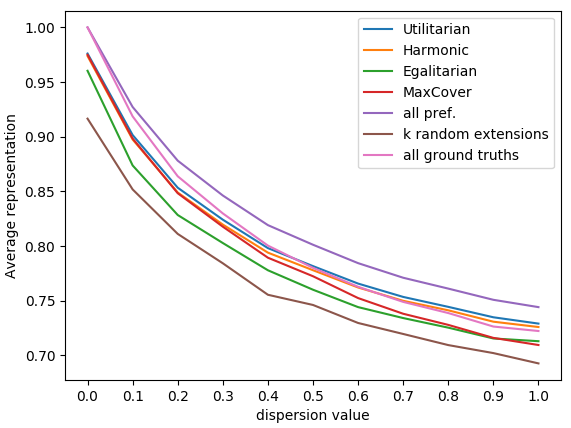}
		\caption{Average representation.}
		\label{fig:avgExp}
	\end{subfigure}
	\begin{subfigure}[b]{0.45\textwidth}
		\centering
		\includegraphics[width=\textwidth]{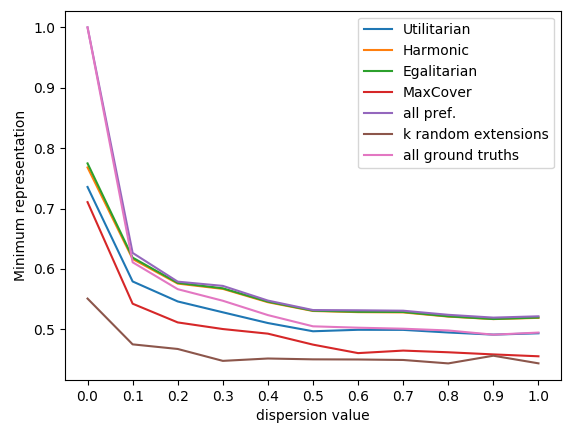}
		\caption{Minimal representation.}
		\label{fig:minExp}
	\end{subfigure}
	\hfill
	\begin{subfigure}[b]{0.45\textwidth}
		\centering
		\includegraphics[width=\textwidth]{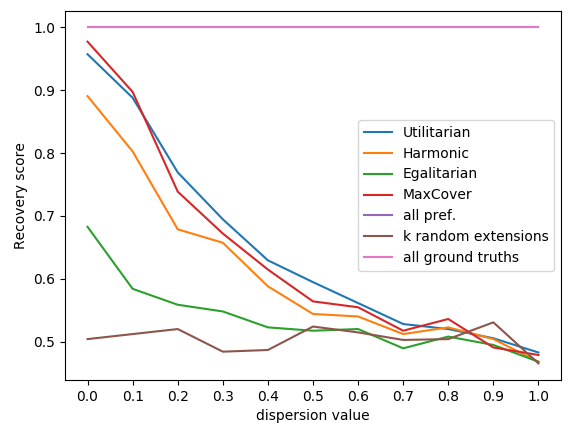}
		\caption{Recovery score.}
		\label{fig:recExp}
	\end{subfigure}
	\caption{Results of the experiments.}
	\label{fig:exps}
\end{figure*}

In this section, we present our experiments. 
We test our rules on a concrete scenario, and 
find that, on our data,  
the Utilitarian and Harmonic rule
achieve good representation, 
and that their greedy versions scale to practically relevant inputs sizes.

\subsubsection*{Setup.} We implemented our framework in Python. To compute the preferred extensions, we use \texttt{ASPARTIX} \citep{egly2008aspartix}. We ran our experiments on a device equipped with two Intel Xeon E5-2650 v4 (12 cores @ 2.20GHz) CPUs and 256GB of RAM.


Since there is no dataset of ABSAFs, we generate our data. Note that we cannot take advantage of datasets consisting of plain-text discussions with user approvals,\footnote{See, e.g., the ``Computational Democracy Project'', from which we took Example~\ref{example:absaf}: \url{https://compdemocracy.org/}} as there are still challenges to be overcome when it comes to the connection between argument mining from text and formal Dung-style argumentation~\citep{BaumannWHHH20}.
Moreover, we also cannot use existing benchmarks of AFs\footnote{See, e.g., \url{http://argumentationcompetition.org/}} (and only generate the votes), as these involve very large AFs. This is beyond our scope, as we want to focus on more moderately-sized discussions. Hence, we generate both the AFs and the ballots.



To generate the frameworks, we use \texttt{AFBenchGen2}~\citep{CeruttiGV2016benchgen}, a tool for generating AFs implementing various graph-generation methods. We use the Barabási-Albert sampling algorithm~\citep{barabasi1999emergence}, which produces scale-free graphs, i.e., graphs whose degree distribution follows a power law. Note that Barabási-Albert graphs are, by default, acyclic; this is unsuitable for AFs. Hence, with \texttt{AFBenchGen2}, users can specify the probability of a node being in a cycle, and the tool will add extra attacks accordingly. We refer to this probability as $p_c$, and to the number of arguments as $n_a$.

In our experiments, we are interested in AFs with a certain number of preferred extensions: on the one hand, we need enough preferred extensions to make the problem of choosing an outcome nontrivial. On the other hand, we need to limit the number of extensions for performance reasons. As there is no default option to generate AFs with a given number of extensions, we generate a large number of AFs (with varying arguments) and then filter out those with a desirable number of extensions.

In particular, we generate AFs with $n_a\in\{10,\ \ldots,\ 50\}$ and $p_c\in\{0.25,\allowbreak 0.5,\allowbreak 0.75\}$. For each combination of these values, we generate $20$ AFs. From this set we then select a certain number of AFs (depending on the experiment) with the desired number of preferred extensions (also depending on the experiment).

To generate votes, we use a mixed Mallows model \citep{mallows1957non, szufa2020drawing}.
Here, there are some central votes (from now on, \emph{ground truths}), of which the ballots are noisy signals. We might think of a ground truth as some rational viewpoint a voter might lean towards. Intuitively, each voter's preference will be a (possibly irrational) approximation of such viewpoints.

Formally, for a ground truth $S$ and distance $d$, the probability of sampling a vote $S^\prime$ is proportional to $\Phi^{d(S,S^\prime)}$, where $\Phi\in[0, 1]$ is called the \emph{dispersion}. Observe that, if $\Phi=0$, all votes have zero distance from the ground truth, while if $\Phi=1$, all votes are equally likely. \citet{howtosample} give a polynomial-time algorithm for sampling approval ballots with this model.

In our scenario, given an AF, we sample the ground truths uniformly at random from its preferred extensions.
Then, for each ground truth, we sample a fixed number of ballots centered around it. This models a scenario where voters are divided into different ``parties'' centered around different rational viewpoints. Given a ground truth $S$ and a vote $S^\prime$, we define the distance $d(S,S^\prime)=|S^\prime|-|S^\prime\cap S|$, i.e., the number of arguments in $S^\prime$ that are not in $S$. Observe that our distance measure reflects our notion of representability: if a voter approves a subset of the ground truth, the distance is zero (as the ground truth fully ``covers'' her approved arguments).

For our experiments, we select $50$ frameworks with $10$ preferred extensions each. 
Then, for each AF $F$, we sample $5$ ground truths from $\pref(F)$, and for each ground truth, we sample $20$ ballots centered around it. We repeat this generation for $\Phi\in\{0,0.1,\ldots,0.9,1\}$, and for each AF and $\Phi$, we generate $5$ ABSAFs, for a total of $50\cdot 11\cdot 5=2750$ ABSAFs.

We ran all experiments for both the base definition of representation (cf.\ Definition~\ref{def:representation}) and core-representation (cf.\ Definition~\ref{def:core-representation}). Here, we show results for the former; the others are comparable, and are included in the appendix.

In our first experiment, given a rule, we compute the score of its selected outcome over all generated AFs and for $k\in\{1,\ldots,7\}$. For each value of $k$, we average the results over each value of the dispersion parameter (i.e., we plot the dispersion parameter $\Phi$ on the $x$-axis, and the function score on the $y$-axis for each $k$). Since there are $5$ ground truths, our expectation (for low dispersion) is that, for $k=5$, our method should be able to represent the voters well. By considering $k<5$ we investigate how robust our approach is (i.e., whether the score dips) when selecting outcomes smaller than the number of ground truths. Conversely, by looking at $k\in\{6,7\}$, we check whether we gain any performance by outcomes larger than the ground truth. The results of this experiment are displayed in Figure~\ref{fig:varyingExp}. We ran this for Utilitarian, Harmonic and Egalitarian. Here we show results for Utilitarian (the others are comparable, and are included in the appendix).

Next, we run the Utilitarian, Egalitarian, Harmonic and \maxCov{} rules on our data with $k=3$. We look at the following metrics, all of which we average over the different dispersion parameters:
\begin{itemize}
  \item The average and minimal representation score (over all voters) given by the selected outcome;
  \item The ``recovery score'' of an outcome computed as $\nicefrac{r(\outcome)}{k}$, where $r(\outcome)$ is the number of ground truths included in $\outcome$. Intuitively, this captures how well a rule can recover the ground truths around which the voters' preferences center;
  \item The approximation ratio of the greedy variant of each rule.
\end{itemize}

We report the results in Figures~\ref{fig:avgExp},~\ref{fig:minExp} and~\ref{fig:recExp}. For the greedy approximation experiment, we do not show the full data here (it can be found in the appendix), but only the minimum of the (average) approximation ratio (across all dispersion parameters): for Utilitarian, Harmonic, Egalitarian and \maxCov{} we have a minimum value of $99.26\%$, $96.61\%$, $89.41\%$ and $98.75\%$, respectively. Observe that we ran all the above also for $k\in\{2,5\}$; the results are comparable to what we show, and are included in the appendix.

\begin{table}[t]
	\centering
	\caption{Performance experiment.}
	\begin{tabular}{c|cccc}
		\toprule
		& \multicolumn{2}{c}{\textbf{non-greedy}} & \multicolumn{2}{c}{\textbf{greedy}} \\ \cmidrule(lr){2-3}\cmidrule(lr){4-5}
		\emph{\# pref. ext.}     & \emph{successes} & \emph{time} &  \emph{successes}  & \emph{time} \\ \midrule
		$\{8, \ldots, 12\}$  & 30/30     & 0.03          & 30/30      & 0.01          \\
		$\{13, \ldots, 17\}$ & 30/30     & 0.73          & 30/30      & 0.02          \\
		$\{18, \ldots, 22\}$ & 30/30     & 5.63          & 30/30      & 0.03          \\
		$\{23, \ldots, 27\}$ & 1/30      & 17.24         & 30/30      & 0.06          \\
		$\{28, \ldots, 32\}$ & 0/30      & -             & 30/30      & 0.12		  \\ \bottomrule
	\end{tabular}
	\label{tab:performance}
\end{table}

Finally, we compare the runtime and scalability of the greedy and non-greedy approaches. To do so, we select AFs with increasing values for $m = |\pref(\af)|$, and compare the performance of the two approaches. We start by selecting $30$ AFs with between $8$ and $12$ preferred extensions each. At the next iteration, we select $30$ AFs that have between $13$ and $17$ preferred extensions, and so on, increasing the minimum and maximum number of preferred extensions allowed by $5$ at every iteration. To sample the voters, for each AF $F$, we choose uniformly at random the dispersion value to use from the set $\Phi\in\{0,0.1,\ldots,0.9,1\}$, and sample approximately $100$ voters centered around $\lfloor \nicefrac{m}{4}\rfloor$ ground truths (chosen uniformly at random from $\pref(F)$).\footnote{Note that the exact number of voters depends on the value of $\lfloor \nicefrac{m}{4}\rfloor$, and in our setup, was never less than $96$. Regardless, the number of voters has a limited impact on the runtime.}

We run the Utilitarian and greedy-Utilitarian rules on each input, looking for outcomes of size $k=\lfloor \nicefrac{m}{4}\rfloor$ and setting a timeout of $45$ seconds. We increase the number of preferred extensions $m$ until no instance out of the $30$ is solved within the time limit. For the instances that terminate before timeout, we report the average runtime. The results are shown in Table~\ref{tab:performance}. We also ran the greedy algorithm on larger AFs, to assess its limits. On a framework with $256$ extensions, we could find an outcome of size $k=64$ in 22 seconds. With $512$ extensions and $k=128$, we exceed the 45-seconds timeout.

\subsubsection*{Discussion} Let us now analyze the results. Looking at Figure~\ref{fig:varyingExp} we see the following:
(1)~As $\Phi$ grows, performance drops at the same rate for all $k$; 
(2)~Our method seems quite robust: if we select outcomes slightly smaller than the true number of ground truths (i.e., $k\in\{3,4\}$), we do not observe substantial performance drops;
(3)~Conversely, if we select outcomes larger than the number of ground truths ($k>5$), we do not gain significant advantages.

Next, looking at Figures~\ref{fig:avgExp},~\ref{fig:minExp} and~\ref{fig:recExp}, we note that all rules outperform the $k$-random-extensions baseline. This is promising, although this baseline achieves quite high performance in some metrics (Figure~\ref{fig:avgExp}). Next, although the ``all preferred extensions'' baseline is (as expected) always best performing, the performance of ``all ground truths'' decreases, being even outperformed by some rules. This is reasonable: as $\Phi$ increases, votes become less similar to the ground truths. Moreover, we can see from Figures~\ref{fig:avgExp} and~\ref{fig:minExp} that Harmonic seems to be a good compromise between Utilitarian and Egalitarian. In particular, Harmonic performs almost as well as Egalitarian in the minimum representation metric. Furthermore, we can see that \maxCov{} performs poorly for minimum representation, and its performance decreases starkly as $\Phi$ increases for average representation. Finally, in Figure~\ref{fig:recExp}, Utilitarian and \maxCov{} seem the best performing rules; however, Harmonic performs comparably.

In general, despite the axiomatic advantages of \maxCov{} (Theorem~\ref{thm:maxCov-JR}),
in our setup, Utilitarian seems to perform better w.r.t. most metrics.\footnote{Observe that
we also tested for an additional metric, namely, the percentage of voters that are 1-represented
by the selected outcome (the plots can be found in the appendix). Clearly, \maxCov{} was the best performing rule, but Utilitarian was
essentially equally good. Therefore, even from the angle of perfect representation, in practice,
Utilitarian seems comparable to \maxCov{}.} Depending on the application scenario, we can recommend
the Utilitarian rule, with the Harmonic rule being a sensible alternative,
compromising between efficiency and fairness.

Finally, we notice that the greedy versions of our rules offer a good approximation ratio, always well above $90\%$ (with the exception of Egalitarian).
Moreover, from Table~\ref{tab:performance} we can see, as hinted by our theoretical findings (Proposition \ref{prop:greedy-normal-complexity}), that the greedy algorithm drastically improves over the performance of the non-greedy one, and seems to scale quite well. These findings combined suggest that the greedy algorithms are a scalable and well-performing method to apply our framework in real-world scenarios.

\section{Conclusion}

\paragraph{Summary} We presented a new framework for representing different
viewpoints in online discussions by combining approval voting
and abstract argumentation.
In this framework, citizens can both propose arguments and vote on
them. We then use argumentation theory to find the maximal, consistent 
and defendable points of view. For smaller instances,
we can then efficiently find a set that represents the \emph{defendable cores}
of all voter's ballots. For larger instances, we propose different methods for picking 
a small set of most representative points of view
and compared them axiomatically and in simulations.
Axiomatically, \maxCov{} showed the best behavior and can be 
recommended when justified representation is desired.
If justified representation is not necessary,
Utilitarian and Harmonic can be recommended according to our experiments, 
with the former being more efficient and the latter being fairer.

\paragraph{Formalizing Discussions} In this paper, we assume discussions to be already formalized as AFs.
In practice, when using real-world data, this formalization is a crucial step. 
In Example~\ref{example:absaf}, we manually created the AF representing the canadian election reform discussion. 
However, 
while manual annotation is feasible for small discussions, a fair moderation is important, as to avoid any subjective biases.
The formalization of larger debates, instead, may require argument mining techniques that can convert (often previously annotated) natural language text into AFs.
While such methods are being worked on~\citep{PalauM09a,CabrioV12,GoffredoCVHS23}, there are still challenges to be overcome~\citep{BaumannWHHH20}.

\paragraph{Future Work}
Our framework could be extended by also allowing disapprovals,
which are commonly seen in practice. Moreover, our general approach is independent 
of the choice of abstract argumentation for identifying consistent sets 
of comments. 
Thus, we plan to investigate the effect of using other mechanisms instead.

\section*{Acknowledgments}

Michael Bernreiter was supported by the Austrian Science Fund (FWF) through project P32830.
Jan Maly was supported by the FWF through project J4581. 
Oliviero Nardi was supported by the European Union's Horizon 2020 research and innovation programme under grant agreement number \includegraphics[height=\fontcharht\font`\B]{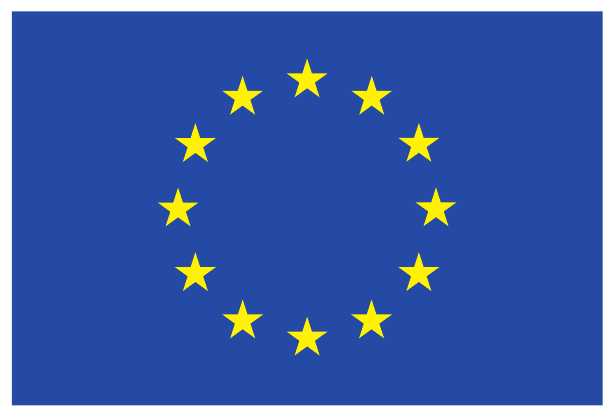}101034440, and by the Vienna Science and Technology Fund (WWTF) through project ICT19-065. 
Stefan Woltran was supported by the FWF through project P32830 and by the WWTF through project ICT19-065. 

\bibliographystyle{apalike}
\bibliography{references}


\clearpage

\appendix

\section{Omitted Plots}

\begin{table}[h]
	\centering
	\caption{Greedy approximation experiment ($k=2$, regular representation).}
	\begin{tabular}{c|cc}
		\toprule 
		\emph{function}     & \emph{min. approx.} & \emph{max. approx.} \\\midrule
		Utilitarian & 99.20\%      & 99.78\% \\
		Harmonic & 95.31\%      & 98.72\% \\
		Egalitarian & 88.42\%      & 96.06\% \\
		\maxCov & 99.09\%      & 99.93\% \\ \bottomrule
	\end{tabular}
\end{table}

\begin{table}[h]
	\centering
	\caption{Greedy approximation experiment ($k=2$, defendable core).}
	\begin{tabular}{c|cc}
		\toprule
		\emph{function}     & \emph{min. approx.} & \emph{max. approx.} \\\midrule
		Utilitarian & 99.20\%      & 99.80\% \\
		Harmonic & 95.31\%      & 98.51\% \\
		Egalitarian & 88.42\%      & 96.12\% \\
		\maxCov & 99.09\%      & 99.50\% \\ \bottomrule
	\end{tabular}
\end{table}

\begin{table}[h]
	\centering
	\caption{Greedy approximation experiment ($k=3$, regular representation).}
	\begin{tabular}{c|cc}
		\toprule
		\emph{function}     & \emph{min. approx.} & \emph{max. approx.} \\\midrule
		Utilitarian & 99.26\%      & 99.77\% \\
		Harmonic & 96.61\%      & 99.28\% \\
		Egalitarian & 89.41\%      & 97.73\% \\
		\maxCov{} & 98.75\%      & 99.84\% \\ \bottomrule
	\end{tabular}
	\label{tab:greedy}
\end{table}

\begin{table}[h]
	\centering
	\caption{Greedy approximation experiment ($k=3$, defendable core).}
	\begin{tabular}{c|cc}
		\toprule
		\emph{function}     & \emph{min. approx.} & \emph{max. approx.} \\\midrule
		Utilitarian & 99.26\%      & 99.79\% \\
		Harmonic & 96.61\%      & 98.73\% \\
		Egalitarian & 89.41\%      & 95.10\% \\
		\maxCov & 98.73\%      & 99.31\% \\ \bottomrule
	\end{tabular}
\end{table}

\begin{table}[h]
	\centering
	\caption{Greedy approximation experiment ($k=5$, regular representation).}
	\begin{tabular}{c|cc}
		\toprule
		\emph{function}     & \emph{min. approx.} & \emph{max. approx.} \\\midrule
		Utilitarian & 99.74\%      & 99.90\% \\
		Harmonic & 94.20\%      & 99.79\% \\
		Egalitarian & 73.96\%      & 97.89\% \\
		\maxCov & 99.43\%      & 99.97\% \\ \bottomrule
	\end{tabular}
\end{table}

\begin{table}[h]
	\centering
	\caption{Greedy approximation experiment ($k=5$, defendable core).}
	\begin{tabular}{c|cc}
		\toprule
		\emph{function}     & \emph{min. approx.} & \emph{max. approx.} \\ \midrule
		Utilitarian & 99.74\%      & 99.89\% \\
		Harmonic & 94.20\%      & 98.99\% \\
		Egalitarian & 73.96\%      & 92.26\% \\
		\maxCov & 99.33\%      & 99.63\% \\ \bottomrule
	\end{tabular}
\end{table}

\begin{figure*}
     \centering
     \begin{subfigure}[b]{0.45\textwidth}
         \centering
         \includegraphics[width=\textwidth]{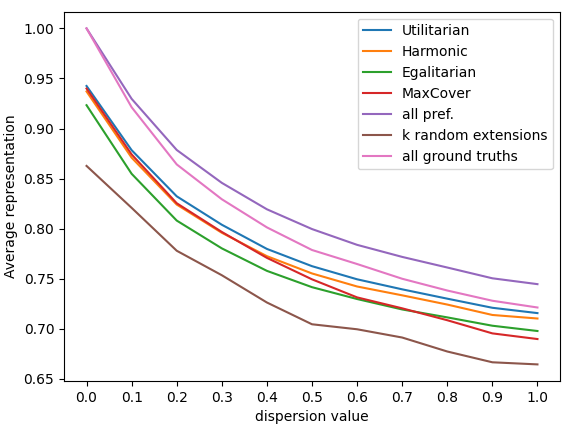}
         \caption{Average representation.}
     \end{subfigure}
      \hfill
     \begin{subfigure}[b]{0.45\textwidth}
         \centering
         \includegraphics[width=\textwidth]{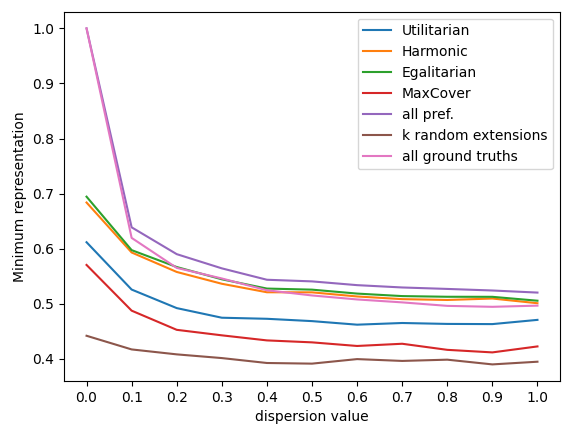}
         \caption{Minimal representation.}
     \end{subfigure}
     \begin{subfigure}[b]{0.45\textwidth}
         \centering
         \includegraphics[width=\textwidth]{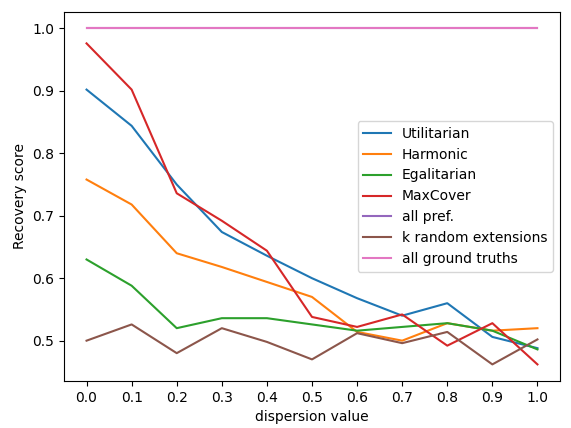}
         \caption{Recovery score.}
     \end{subfigure}
      \caption{Results of the experiments with $k=2$ (Regular representation).}
\end{figure*}

\begin{figure*}
     \centering
     \begin{subfigure}[b]{0.45\textwidth}
         \centering
         \includegraphics[width=\textwidth]{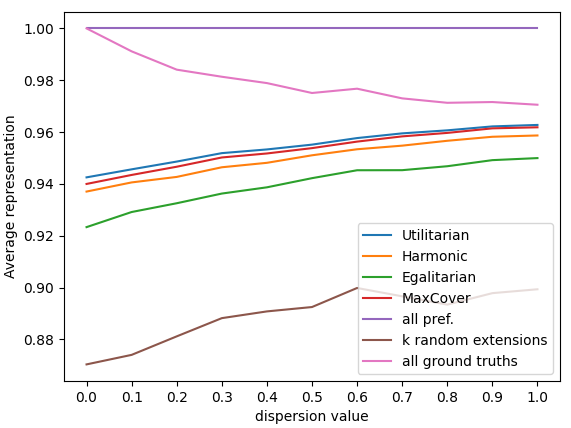}
         \caption{Average representation.}
     \end{subfigure}
      \hfill
     \begin{subfigure}[b]{0.45\textwidth}
         \centering
         \includegraphics[width=\textwidth]{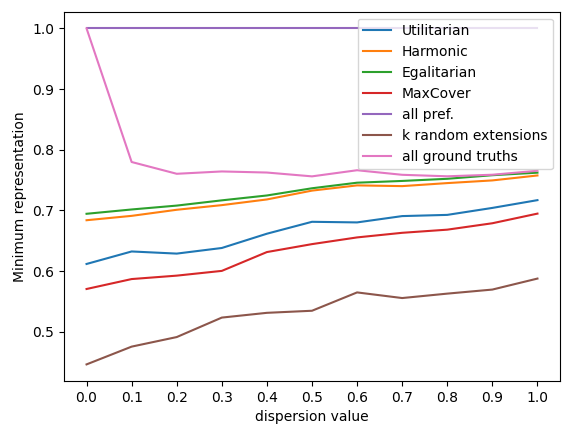}
         \caption{Minimal representation.}
     \end{subfigure}
     \begin{subfigure}[b]{0.45\textwidth}
         \centering
         \includegraphics[width=\textwidth]{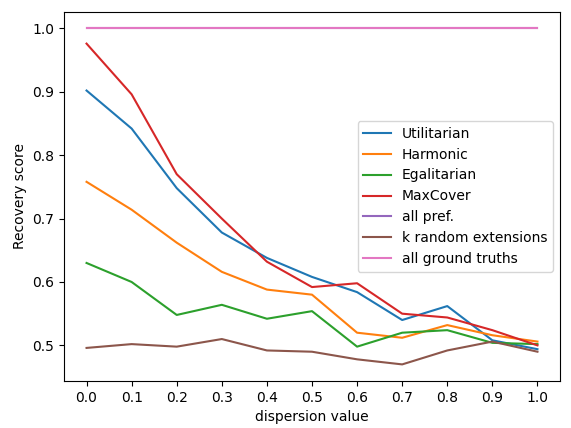}
         \caption{Recovery score.}
     \end{subfigure}
      \caption{Results of the experiments with $k=2$ (Defendable core).}
\end{figure*}

\begin{figure*}
     \centering
     \begin{subfigure}[b]{0.45\textwidth}
         \centering
         \includegraphics[width=\textwidth]{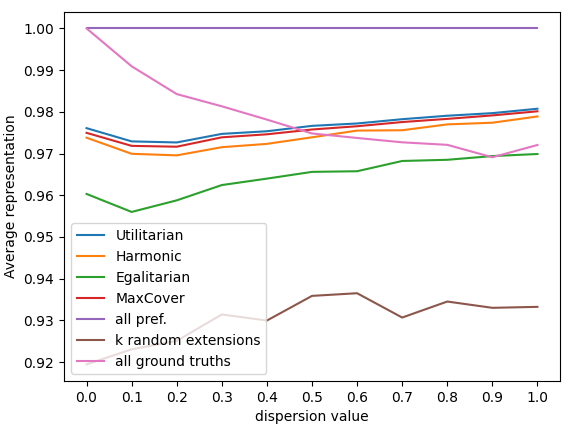}
         \caption{Average representation.}
     \end{subfigure}
     \hfill
     \begin{subfigure}[b]{0.45\textwidth}
         \centering
         \includegraphics[width=\textwidth]{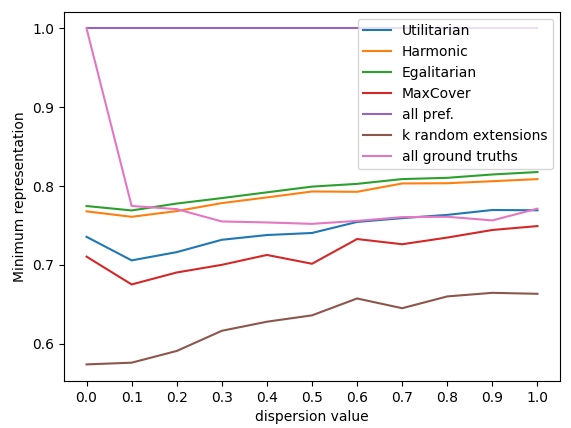}
         \caption{Minimal representation.}
     \end{subfigure}
     \begin{subfigure}[b]{0.45\textwidth}
         \centering
         \includegraphics[width=\textwidth]{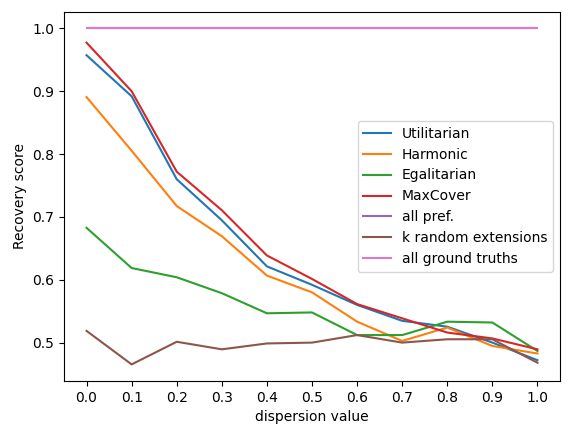}
         \caption{Recovery score.}
     \end{subfigure}
      \caption{Results of the experiments with $k=3$ (Defendable core).}
\end{figure*}

\begin{figure*}
     \centering
     \begin{subfigure}[b]{0.45\textwidth}
         \centering
         \includegraphics[width=\textwidth]{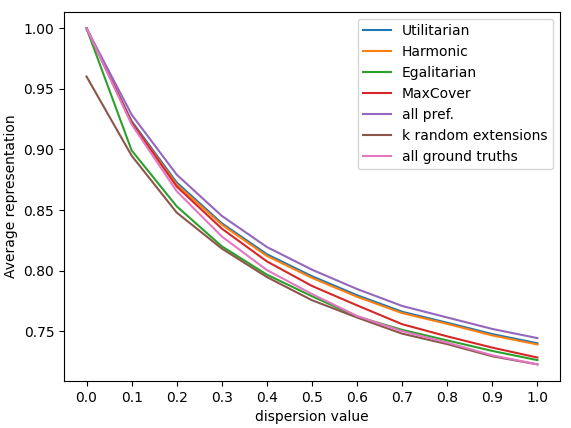}
         \caption{Average representation.}
     \end{subfigure}
     \hfill
     \begin{subfigure}[b]{0.45\textwidth}
         \centering
         \includegraphics[width=\textwidth]{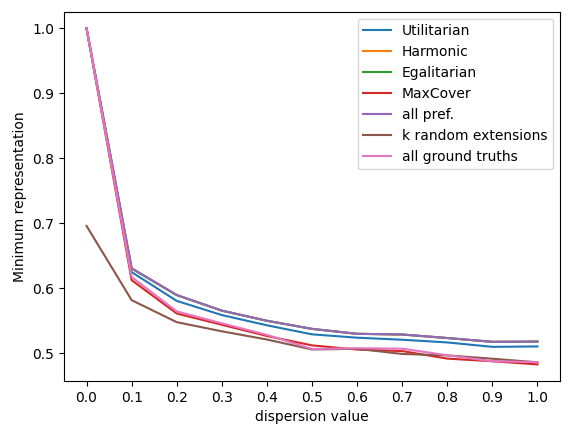}
         \caption{Minimal representation.}
     \end{subfigure}
     \begin{subfigure}[b]{0.45\textwidth}
         \centering
         \includegraphics[width=\textwidth]{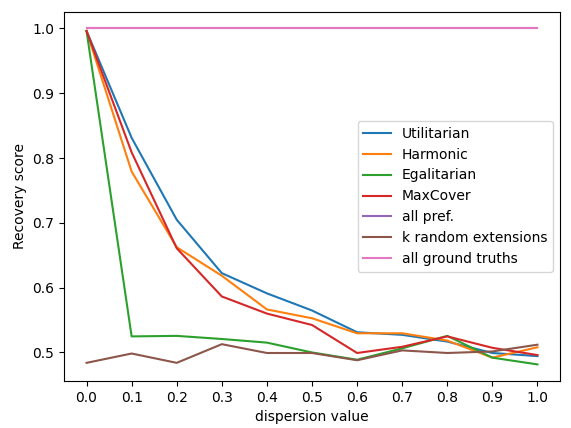}
         \caption{Recovery score.}
     \end{subfigure}
      \caption{Results of the experiments with $k=5$ (Regular representation).}
\end{figure*}

\begin{figure*}
     \centering
     \begin{subfigure}[b]{0.45\textwidth}
         \centering
         \includegraphics[width=\textwidth]{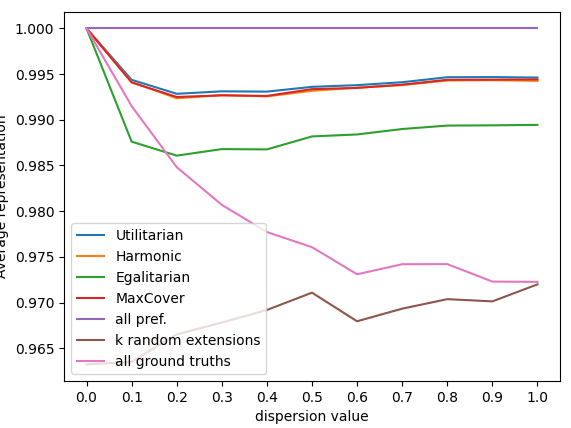}
         \caption{Average representation.}
     \end{subfigure}
     \hfill
     \begin{subfigure}[b]{0.45\textwidth}
         \centering
         \includegraphics[width=\textwidth]{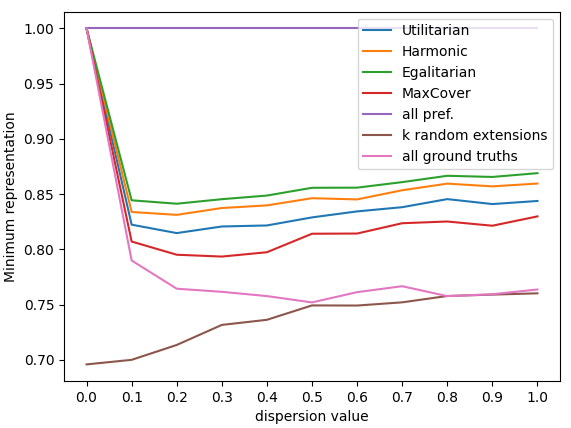}
         \caption{Minimal representation.}
     \end{subfigure}
     \begin{subfigure}[b]{0.45\textwidth}
         \centering
         \includegraphics[width=\textwidth]{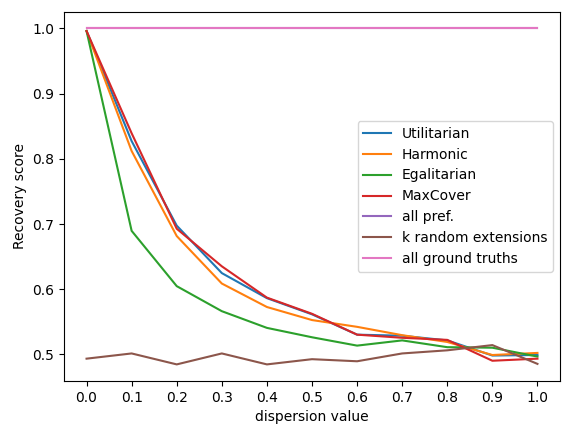}
         \caption{Recovery score.}
     \end{subfigure}
      \caption{Results of the experiments with $k=5$ (Defendable core).}
\end{figure*}

\begin{figure*}
     \centering
     \begin{subfigure}[b]{0.45\textwidth}
         \centering
         \includegraphics[width=\textwidth]{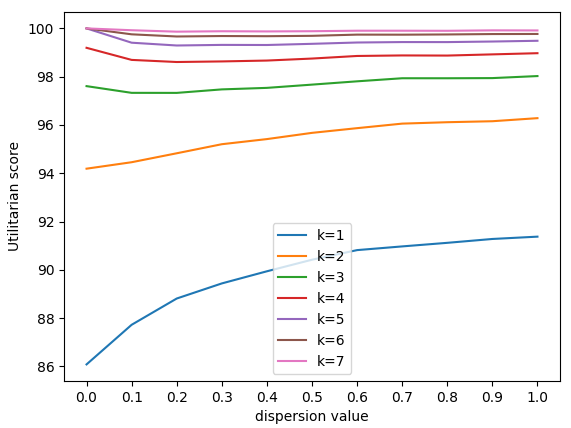}
         \caption{Utilitarian, core-representation.}
     \end{subfigure}\\
     \begin{subfigure}[b]{0.45\textwidth}
         \centering
         \includegraphics[width=\textwidth]{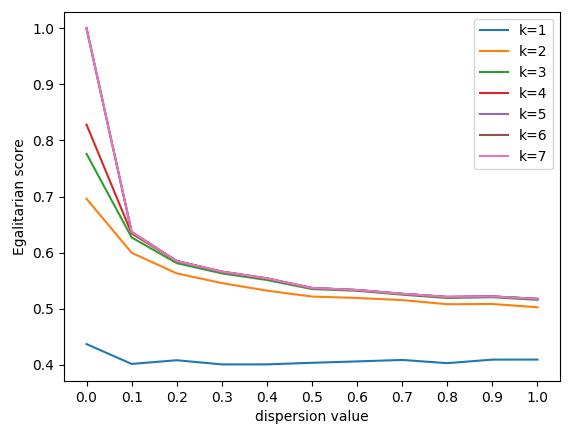}
         \caption{Egalitarian, regular representation.}
     \end{subfigure}
     \hfill
     \begin{subfigure}[b]{0.45\textwidth}
         \centering
         \includegraphics[width=\textwidth]{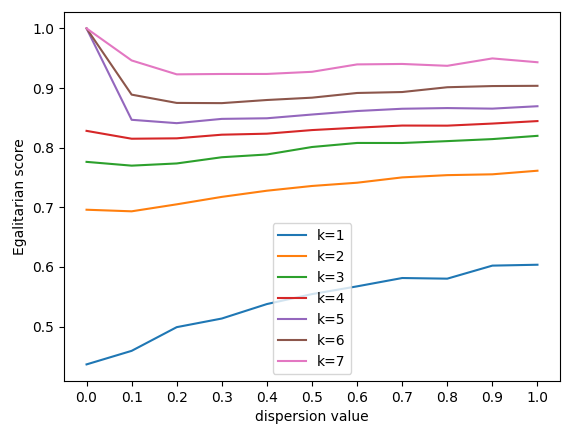}
         \caption{Egalitarian, regular representation.}
     \end{subfigure}
     \begin{subfigure}[b]{0.45\textwidth}
         \centering
         \includegraphics[width=\textwidth]{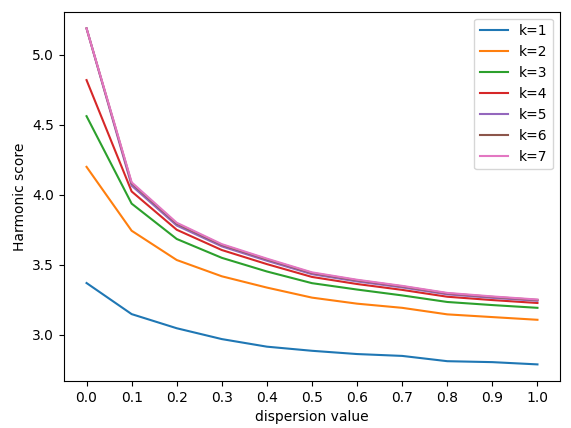}
         \caption{Harmonic, core-representation.}
     \end{subfigure}
     \hfill
     \begin{subfigure}[b]{0.45\textwidth}
         \centering
         \includegraphics[width=\textwidth]{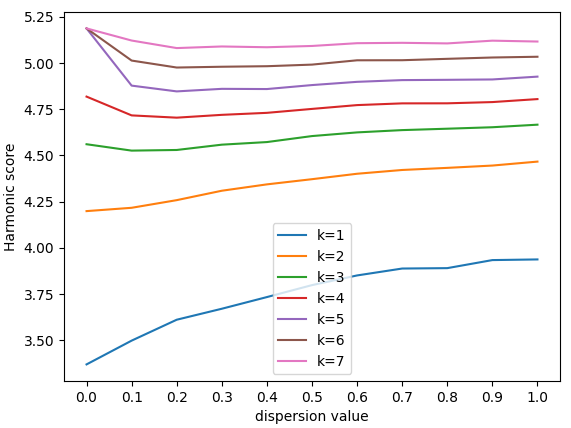}
         \caption{Harmonic, core-representation.}
     \end{subfigure}
      \caption{Results of the experiments with variable $k$.}
\end{figure*}

\begin{figure*}
     \centering
     \begin{subfigure}[b]{0.45\textwidth}
         \centering
         \includegraphics[width=\textwidth]{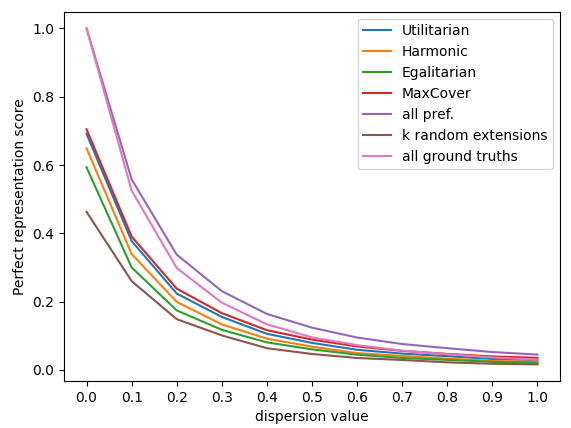}
         \caption{$k=2$, regular representation.}
     \end{subfigure}
     \hfill
     \begin{subfigure}[b]{0.45\textwidth}
         \centering
         \includegraphics[width=\textwidth]{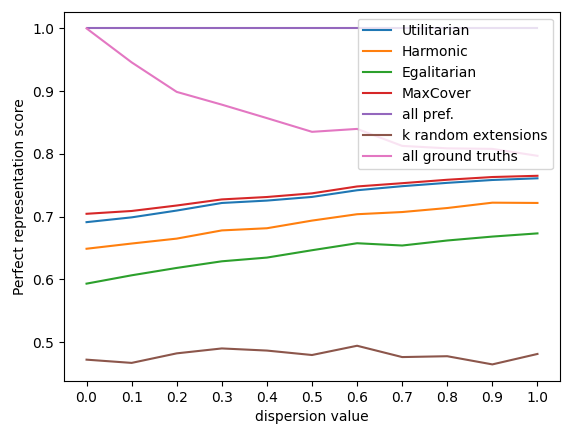}
         \caption{$k=2$, core-representation.}
     \end{subfigure}
     \begin{subfigure}[b]{0.45\textwidth}
         \centering
         \includegraphics[width=\textwidth]{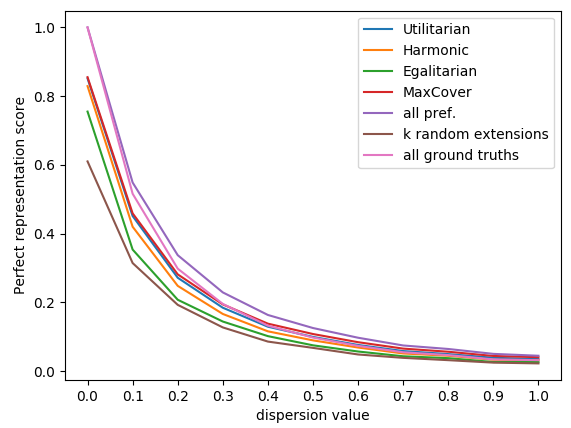}
         \caption{$k=3$, regular representation.}
     \end{subfigure}
     \hfill
     \begin{subfigure}[b]{0.45\textwidth}
         \centering
         \includegraphics[width=\textwidth]{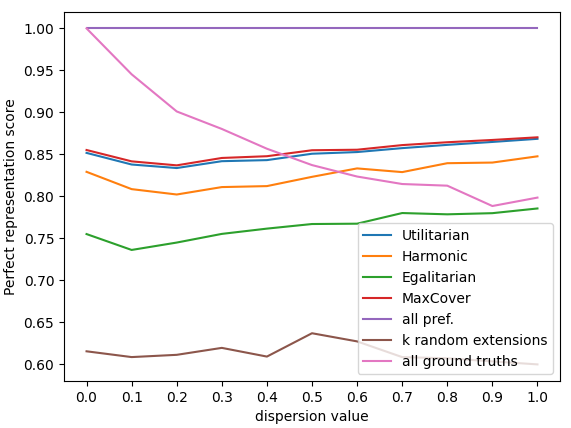}
         \caption{$k=3$, core-representation.}
     \end{subfigure}
     \begin{subfigure}[b]{0.45\textwidth}
         \centering
         \includegraphics[width=\textwidth]{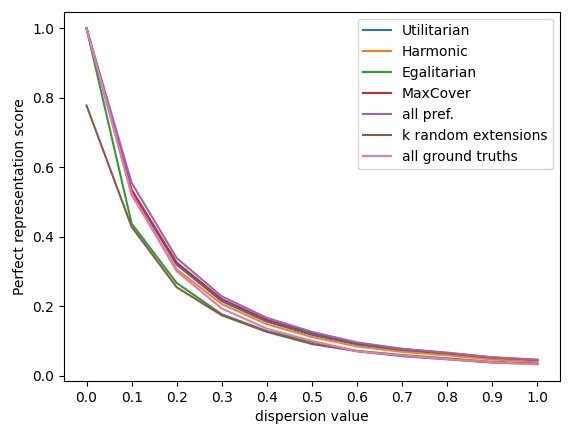}
         \caption{$k=5$, regular representation.}
     \end{subfigure}
     \hfill
     \begin{subfigure}[b]{0.45\textwidth}
         \centering
         \includegraphics[width=\textwidth]{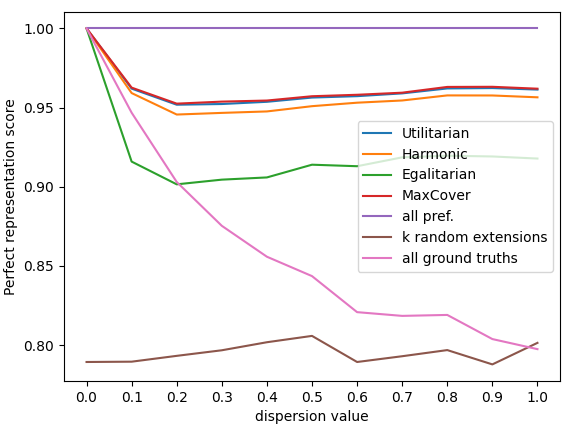}
         \caption{$k=5$, core-representation.}
     \end{subfigure}
      \caption{Results of the ``perfect representation score'' experiments.}
\end{figure*}


\end{document}